\documentclass[num-refs]{wiley-article}

\usepackage{amsmath, amssymb, amsfonts}
\usepackage{graphicx,psfrag,epsf}
\usepackage{enumerate}
\usepackage{natbib}
\usepackage{url} 
\usepackage{xr}
\usepackage{xspace}
\usepackage{booktabs}
\usepackage{mathtools}
\usepackage{tabularx} 
\usepackage{hyperref}
\usepackage{multirow}
\usepackage{algorithm}
\usepackage{algorithmic}
\usepackage{subcaption}
\usepackage{color}
\usepackage{booktabs}
\usepackage{bm, stmaryrd}
\usepackage{comment}
\usepackage{verbatim}
\usepackage{url}
\newcommand{\bas}[1]{\begin{align*}#1\end{align*}}
\newcommand{\ba}[1]{\begin{align}#1\end{align}}
\newcommand{\cN}{\mathcal{N}}
\newcommand{\bfeps}{\boldsymbol{\epsilon}}
\newcommand{\ttvec}{\text{vec}}

\newcommand{\sign}{\text{sign}}
\newcommand{\ttspan}{\text{span}}
\newcommand{\beq}[1]{\begin{equation}#1\end{equation}}
\newcommand{\bsplt}[1]{\begin{split}#1\end{split}}
\newcommand{\bR}{\mathbb{R}} 
\newcommand{\tr}{\text{trace}}
\newcommand{\that}{\hat{\theta}}
\newcommand{\tbar}{\bar{\theta}}

\usepackage{siunitx}

\renewcommand{\thesection}{\Alph{section}}

\papertype{Original Article}
\paperfield{Journal Section}

\title{Graph-Fused Multivariate Regression via Total Variation Regularization}

\author[1\authfn{1}]{Ying Liu, Ph.D.}
\author[2\authfn{2}]{Bowei Yan, Ph.D.}
\author[3\authfn{3}]{Kathleen Merikangas, Ph.D.}
\author[4]{Haochang Shou, Ph.D.}

\affil[1]{Mental Health Data Science, the New York State Psychiatric Institute and the Department of Psychiatry at Columbia University Irving Medical Center, New York, NY 10032}
\affil[3]{Genetic Epidemiology Research Branch,National Institute of Mental Health, Bethesda, MD 20892}
\affil[4]{Department of Biostatistics, Epidemiology and Informatics, University of Pennsylvania, Philadelphia, PA 19104}
\corraddress{Ying Liu, PhD, Department, Institution, City, State or Province, Postal Code, Country}
\corremail{ying.liu@nyspi.columbia.edu}

\runningauthor{Liu et al.}

\begin{document}

\maketitle

\begin{abstract}
In this paper, we propose {\it Graph-Fused Multivariate Regression} (GFMR) via Total Variation regularization, a novel method for estimating the association between a one-dimensional or multidimensional array outcome and scalar predictors. While we were motivated by data from neuroimaging and physical activity tracking, the methodology is designed and presented in a generalizable format and is applicable to many other areas of scientific research.  The estimator is the solution of a penalized regression problem where the objective is the sum of square error plus a total variation (TV) regularization on the predicted mean across all subjects.  We propose an algorithm for parameter estimation, which is efficient and scalable in a distributed computing platform. Proof of the algorithm convergence is provided, and the statistical consistency of the estimator is presented via an oracle inequality.  We present 1D and 2D simulation results and demonstrate that GFMR outperforms existing methods in most cases. We also demonstrate the general applicability of the method by two real data examples, including the analysis of the 1D accelerometry subsample of a large community-based study for mood disorders and the analysis of the 3D MRI data from the attention-deficient/hyperactive deficient (ADHD) 200 consortium. 

\end{abstract}

\section{Introduction}
\label{sec:intro}
Modern biomedical research continues to generate outcome data as a high-dimensional and multivariate object with complex spatial and temporal structures. For example, wearable sensor devices such as accelerometers and smart phones have been increasingly used in observational studies such as National Health and Nutrition Examination Survey (NHANES) and UK Biobank \citep{naslund2015feasibility,pantelopoulos2010survey,patel2015wearable} to densely and objectively track subjects' physical and mental health. Such measures have been shown to avoid self-reporting bias compared with the traditional survey-based methods and to allow continuous monitoring of biosignals such as physical activity patterns over an extended time period and more sensitive to changes with health-related behaviors and clinical outcomes. In neuroscience, sustained growth in computing power and the increasing affordability to collect high-throughput multimodal medical images such as structural and functional magnetic resonance imaging (MRI) has brought promises to search for etiologic evidence and prognostic markers for human diseases. Appropriate statistical methods that acknowledge the multidimensionality of the outcome and the dependency structures induced by temporal or spatial continuity are needed. Motivated by two specific studies that collect 1D time series of physical activity data and 3D neuroimaging data as the outcome measures, we aim to propose a generalizable regression framework that simultaneously reconstructed the denoised outcomes.

In particular, we focus on the problem of regressing the multiple dimensional arrays on a set of scalar features and propose to use graph-guided total variation penalty to ensure temporal or spatial continuity in the outcome. We propose a Graph-fused Multivariate Regression (GFMR) methods, which extend the Total Variation (TV) denoising, a single image smoothing method, to solve regression problem with high dimensional outcomes and produce smooth coefficient maps.
Total Variation smoothing method recovers blurred signals or images by penalizing the average intensity difference of adjacent records/pixels using a $\ell_1$ norm. \cite{rudin1992nonlinear,steidl2006splines} generalized the approach to penalizing higher-order gradient differences and pointed out a connection between the dual formulation of TV and spline smoothing of images. TV denoising can be viewed as a support vector regression problem in the discrete counterpart of the Sobolev space, and can be implemented by efficient algorithms (see e.g., \cite{beck2009fast, condat2013direct,padilla2016dfs}).

There are two ways to impose TV penalization. One way is to assume that the outcomes (i.e., the images) have sparse non-smooth points (the discontinuous points for the step-wise constant approximation for the mean of outcome) , and the penalty is posed upon the change of values of the neighboring elements. Alternatively, one could impose a penalty on the regression coefficients and assume that the associations change smoothly. 
In this paper, we propose a method based on the former assumption, where we introduce regularization directly on the outcome data. Thanks to our choice of the objective function, we can develop an efficient distributed algorithm that converges to the global optimizer. Although we assume the smoothness is in the outcome space, the estimation for coefficients are also smooth as a linear transformation of the fitted outcomes (see Lemma A). 
Alternatively, \cite{chen2016local} introduced a local region image-on-scalar regression method, which used the $\ell_1$ regularization on the coefficients instead of the outcome. Despite the fact, their effort to derive a scalable algorithm, their proposed algorithm based on ADMM cannot handle 2D or 3D outcomes directly. It needs to slice the 3D image into small 2D pieces and solve the penalized regression separately. Implementation is thus case-specific. The final estimation is based on combined local estimators \cite{chen2016local}; thus, it is hard to derive the theory for global convergence and consistency. In comparison, We demonstrate that our proposed algorithm can be used directly to outcomes of 3D brain imaging.

In a nutshell, we transformed the original objective function and proposed an over-relaxed ADMM algorithm. A sub-step of the iterative updates in the proposed algorithm is to solve a graph-fused lasso (GFL) problem. This sub-step is the most computationally intensive step, but by our choice of regularizing the outcome instead of the coefficient, this step has a separable objective and can be computed distributively. And we can solve this step efficiently, borrowing the strength of the state-of-art algorithm for GFL. Our proposed algorithm also inherits the appealing feature of GFL that our method can also easily handle complex adjacency structure. The smoothness constraint can go beyond adjacent voxels and be customized by any adjacency structures corresponding to the edges from an arbitrary graph.  This could include, for example, the periodicity in longitudinal data such as seasonal effects in physical activity or temporal brain activation during a task fMRI scan.

To the best of our knowledge, our proposed method is the first to provide a global coefficient map estimator for the multivariate outcome, smoothed with respect to a flexible user-defined adjacency structure.  In Section \ref{sec:algorithm}, we provided a novel algorithm that our analysis proved, is globally convergent. We develop the consistency of the estimator via an oracle inequality under the Gaussian noise assumption in Section \ref{sec:consistency}.
Our simulation results in Section \ref{sec:simu}, and 1D and 3D real data examples in Sections 7, 8 demonstrate the feasibility and competitive performance of the proposed algorithm to be applied to regression problems involving high dimensional activity or brain imaging data. 

\section{Other Relevant Works}
Total variation penalization in the regression setting has mostly been used to predict a scalar outcome from multivariate or high-dimensional predictors. For example,  fused lasso methods \citep{tibshirani2005, tib2011} were proposed to predict a scalar outcome and assumed that the coefficients of certain predictors are similar to each other. 
$L_1$ penalization was also proposed for scalar-on-image regression \citep{wang2018}. GraphNet \citep{logan2013} used the elastic net penalization for smooth scalar-on-image regression.

There have been relatively few works on regression with the multivariate outcome, and even fewer that acknowledge the dependency structures in the outcome. One line of research treats the multivariate outcomes as random realizations of functional objects and conduct function-on-scalar regression (FoSR) \citep{reiss2010fast,goldsmith2016assessing,scheipl2015functional}.
FoSR assumes both the outcomes and the regression coefficients are smooth over time or space, as represented by the expansion of a set of pre-specified basis functions.
There also exist methods for tensor response regression, such as the envelope regression ~\citep{zhou2013tensor,li2017parsimonious, Cook2010}. These methods take advantage of inter-relationship within the response variables to improve prediction accuracy under general sparsity assumption and underlying lower-dimensional representation. However, they have been demonstrated in 2D or 3D applications, but has not been proposed and evaluated for 1D outcomes.

\section{Notations and Operations}
Multidimensional array $A \in \bR ^{r_1 \times \dots \times r_m}$ is called an m{\it th}-order-tensor. The $\ttvec(\cdot)$ operator denotes the vectorization operation that stack the entries of a tensor into a one-dimensional vector, so that an entry $a_{i_1,\dots,i_m}$ of $A$ maps to the $j$th entry of $\ttvec(\cdot)$, where $j=1+\sum_{k=1}^m(i_k-1)\Pi_{k'=1}^{k-1}r_k'$, and $\text{mat}(\cdot)$  denotes the inverse operator. For example, when the tensor is 2D, vectorization is to stack the columns of the matrix. 

 We will use $\|\cdot\|_F$ for Frobenius norm of a tensor,
$\|A\|^2_F= \sum_{i_1,\dots, i_m}A_{i_1,\dots, i_m}^2$. We denote the $\ell_1$ norm of a tensor as the $\ell_1$ norm of its vectorization: $\|A\|_{\ell_1}= \|\ttvec(A)\|_{\ell_1}$. $\otimes$ represents the Kronecker product between two tensors.  For an integer $n$, denote $[n]=\{1,2,\dots,n\}$.

\section{Model and Inference}
\label{sec:model}
\subsection{Regression Model}
Suppose we observe the covariates and tensor outcomes from $n$ subjects $(X_i,\bar{Y}_i)_{i=1}^n$, where $X_i\in \bR^{p}, \bar{Y}_i\in \bR^{r_1 \times \dots \times r_m}$, $M=r_1 r_2 \dots r_m$ is the total number of outcome entries and $Y_i=\ttvec({\bar{Y}_i}) \in \bR^{M}$.
Consider the following model 
\bas{
Y_{i} =& X_i^T\Gamma + \epsilon_{i} =\sum_{t \in [p]} X_{it}\Gamma_{t\cdot}+\epsilon_i, 
}
where $\Gamma\in \bR^{p\times M}$ is the coefficient matrix of interest, and each row $\Gamma_{t \bullet}\in \bR^{M}$ is a tensor of the same shape as $Y_i$'s, representing the coefficient map corresponds to the $t${\it th} feature. 

While our theory is derived under the case of i.i.d. Gaussian noise, the method can be applied with any real-value responses without distribution assumptions to produce smooth coefficient maps. 
Let $X\in \bR^{n\times p}$ be the design matrix and $Y\in \bR^{n\times M}$ be the matrix of outcomes.
Now we can stack all the observations and reformulate the problem into the following matrix form:
\ba{
\ttvec(Y^T) = \ttvec((X\Gamma)^T)+\bfeps, 
\label{model}
}
where 
$\ttvec(Y^T)=\left(\begin{array}{c}
Y_{1}\\Y_{2}\\\vdots\\Y_{n}\\
\end{array}\right) \in \mathbb{R}^{nM} $ is a long vector. 

Define the hat matrix as $H:=X(X^TX)^{-1}X^T$, which is the projection matrix to $\text{span}(X)$. And the projection matrix to $\text{span}(X)^{\perp}$ is $I-H$. In the vectorization form, we can define the extended projection matrix to project measurements from each voxel to the linear space $\text{span}(X)^{\perp}$:
\bas{
I_{nM}-H_v = (I-H) \otimes I_M\in \mathbb{R}^{nM\times nM},
}
where $H_v$ is the projection matrix to linear space $\text{span}(X)$.
Notice that  for any $\Gamma  \in \mathbb{R}^{p\times M}$,
\bas{
(I-H_v) \ttvec((X\Gamma)^T) =& ((I-H) \otimes I_M )\ttvec(\Gamma^TX^T)= \ttvec(I_M\Gamma^TX^T (I-H)^T)=\bold{0}
}

The penalization for a smoother coefficient map is introduced through an undirected graph.
Let $D$ be the incidence matrix for the graph $G$ whose edges represent the smoothing affinity. For example, the graph $G$'s nodes can be all the pixels/voxels in a image or timestamps in 1D functional data.  The grid graph is the most commonly used graph, where the edges are between all the adjacent pixels in image or timestamps in 1-D data. The graph $G$ can also be chosen to reflect some sophisticated affinity relationships.

To be specific, let $M$ and $m$ be the number of vertices and edges respectively, $D\in \bR^{M \times m}$, where the $j$th column in $D$ corresponds to $j$th edge $e_j$. And $D_{i,j}=0$, for $i$ such that $v_i$ is not vertex of $e_j$. For vertexes of $e_j$, $\{i_1,i_2\},i_1<i_2$, $D_{i_1,j}=1$ and $D_{i_2,j}=-1$.
Note the objective function has a $\ell_1$ penalization term for the differences of means on vertexes of the edges. The signs of $D_{i_1,j}$ and $D_{i_2,j}$ can be switched without changing the objective.
Define the extended incidence matrix $D_v=I_n\otimes D \in \bR^{nM\times nm}$, then
$ \| X\Gamma D\|_{\ell_1}=\sum_{i=1}^n \|X_i\Gamma D\|_{\ell_1}$.
The difference in the dimensions in outcome is thus taken care of in defining the incidence matrix. The algorithm is the same for 1D, 2D, or higher dimensional algorithm. 
We propose the Graph-Fused Multivariate Regression (GFMR) by minimizing the loss function as follows.
\begin{equation}
\min_\Gamma \frac{1}{2}\| Y-X\Gamma\|_F^2 +\lambda \| X\Gamma D\|_{\ell_1}.
\label{eq:obj}
\tag{P}
\end{equation}
where $\lambda$ is a tuning parameter controlling the trade-offs between the linear correlation and the smoothness of the fitted image.

\subsection{Algorithm}
\label{sec:algorithm}
Our algorithm is based on the Alternating Direction Method of Multipliers (ADMM), a general algorithm for distributed computing (see~\citet{boyd2011distributed} for an overview). However, to implement ADMM naively to our proposed objective function~\eqref{eq:obj} suffer from no convergence due to the high dimension of the outcome. Thus we propose the following transformation of the objective function and an over-relaxed ADMM algorithm, which we will later provide proof of convergence.

To make our problem satisfy the {\it TV-separability condition} \citep{wahlberg2012admm}, we consider a reformulation of the original problem. Define $\theta=X\Gamma\in \bR^{n M}$, the problem is equivalent to,
\beq{
\bsplt{
\min_\theta \quad & \frac{1}{2}\| \ttvec(Y^T) -\theta\|^2_F +\lambda \| D_v^T\theta \|_{\ell_1};\\
\textrm{s.t.} \quad & (I-H_v)\theta=0.
}
\label{eq:cons_opt}\tag{CP}
}
The solution of \eqref{eq:obj} and that of \eqref{eq:cons_opt} have the following relationship.
\begin{lemma}
\label{lem:theta2gamma}
Let $\that$ be the optimal solution of \eqref{eq:cons_opt}, and $\hat{\Gamma}$ be the optimal solution of \eqref{eq:obj}. If rank($X)=p$, then 
$
\hat{\Gamma}=(X^TX)^{-1}X^T \text{mat}(\that)_{n\times M}
$.
\end{lemma}
The estimated coefficient map $\hat{\Gamma}$ is thus a linear transformation of the fitted outcomes $\hat{\theta}$, and inherited its smoothness via our proposed $\ell_1$ penalty. 

Now to solve \eqref{eq:cons_opt}, we introduce two auxiliary variables and write the original problem in the following form.
\ba{
\min & \quad \frac{1}{2}\|y-\theta\|^2+\delta((I-H_v)\eta=0)+\lambda \|D_v^T\mu\|_{\ell_1} \nonumber \\
s.t. & \quad \theta=\mu, \theta=\eta \label{eq:augmented_admm}
}
where $y=\ttvec(Y^T)$ and $\delta(\cdot)$ is the characteristic function, which takes $0$ if the condition in the parenthesis is satisfied and infinity otherwise. 
The objective function is separable for $\theta, \eta$ and $\mu$, we solve it with the ADMM as summarized in Algorithm~\ref{alg:palm}.
\begin{algorithm}[h]
\caption{ADMM for Total Variation Regularized Tensor-on-scalar Regression}
\label{alg:alm}
\begin{algorithmic}[1]
\STATE {\bf Input}: Design matrix $X$, Images $Y$, 
tuning parameter $\rho$, error tolerance $tol$.
\STATE Initialize  $\theta^{(0)}, \mu^{(0)}, \eta^{(0)}$ randomly; $U^{(0)}=V^{(0)}=\bm{0}_{nM}$; $k=0$; 
\WHILE{not converge}
\STATE $\theta^{(k+1)}=(\textrm{vec} (Y)+\rho(\eta^{(k)}-U^{(k)}+\mu^{(k)}-V^{(k)}))/(2\rho+1)$;
\STATE $\eta^{(k+1)}= H_v(\theta^{(k+1)}+U^{(k)})$;
\STATE $\mu^{(k+1)} = \arg\min_\mu \lambda\|D_v^T\mu^{(k)}\|_{\ell_1}+\frac{\rho}{2} \|\theta^{(k+1)} +V^{(k)}-\mu^{(k)}\|_F^2$;
\STATE $U^{(k+1)} = U^{(k)} +\theta^{(k+1)} -\eta^{(k+1)}$;
\STATE $V^{(k+1)} = V^{(k)}+\theta^{(k+1)} - \mu^{(k+1)}$;
\STATE $k=k+1$;
\STATE converge {\bf if} $\max\{ \|\theta^{(k+1)}) - \theta^{(k)})\|_F/\|\theta^{(k)})\|_F, \|H_v\theta^{(k+1)}\|_F\} < tol$.
\ENDWHILE
\STATE {\bf Output}: $\Gamma=(X^TX)^{-1}X^T\text{mat}(\theta^{(k+1)})_{n\times M}$;
\end{algorithmic}
\label{alg:palm}
\end{algorithm}

In Algorithm~\ref{alg:palm}, lines 4-6 are the primal variable updates, and lines 7-8 are the dual variable updates. 
The update of $\theta$ at line 4 and the projection step at line 5 both have analytical form and only consist of matrix multiplication. The computational challenge mainly comes from the subproblem at line 6, which can be recognized as a graph-fused lasso (GFL) problem that takes observation $\theta_{k+1}+V_k$, regularization parameter $\lambda/\rho$, and graph incidence matrix $D_v^T$.

Graph-fused lasso(GFL) is a generalization of the fused lasso. It penalizes the first differences of the signal across edges of an arbitrary graph and can be efficiently solved with many off-the-shelf GFL solvers. 
Fused lasso problem with 1-D chained adjacency structure can be solved in $O(n)$ time by the `taut string" algorithm derived by \cite{davies2001local} and the dynamic programming-based algorithm by \cite{johnson2013dynamic}. \cite{Tansey2017} provides an algorithm to solve graph-fused lasso with an arbitrary adjacency graph by decomposition the graph into trials, the problem is solved separately by fast 1D fused lasso solvers for each trial. In our experiments, we implemented the python package provided by \cite{Tansey2017}.

The convergence of ADMM does not have a simple and unified answer. We present the convergence of Algorithm~\ref{alg:palm} in the following theorem. 
\begin{theorem}
Algorithm~\ref{alg:palm} converges linearly to the unique global optimum of problem~\eqref{eq:cons_opt}.
\label{thm:admm_converge}
\end{theorem}
To justify the above theorem, notice our transformed Problem \label{eq:augmented_admm} is a two-block problem; we can thus apply the recent results in general two-block problems~\cite{nishihara2015general} to complete the proof. The detailed proof is deferred to supplementary materials.

\subsection{Software}
We provide the open source python code for implementation of Algorithm~\ref{alg:palm} on github  \url{https://github.com/summeryingliu/imagereg}.

\section{Theoretical Results}
\label{sec:consistency}

In this section, we analyze the statistical property of the estimator given by \eqref{eq:obj}. We present the result via an oracle inequality on the prediction error of the regression. The following two quantities were related to our proof, which is widely used in the analysis of sparse recovery problems.
\begin{definition}[Compatibility factor]
Let $D\in \bR^{M \times m}$ be an incidence matrix. The \it{compatibility factor} of $D$ for a set $\emptyset \subsetneq T\subset [m]$ is defined as
\bas{
\kappa_T := \inf_{\theta \in \bR^T} \frac{\sqrt{|T|}\|\theta\|}{\|(\theta D)_T\|_{\ell_1}};\quad \kappa=\inf_{T\subset [m]} \kappa_T
}
\vskip -0.05in
\label{def:compat_factor}
\end{definition}
The compatibility factor gets its name based on the idea that, on the subset of edges indicated by the set $T$, we require the $\ell_1$-norm and the $\ell_2$-norm to be somehow compatible. Compared with other conditions used to derive sparsity oracle inequalities, such as restricted eigenvalue conditions or irrepresentable conditions, the compatibility factor greater than 0 is shown to be weaker by \cite{van2009conditions}. More discussion about the relationship between different conditions can be found in \cite{van2009conditions}. For graphs with bounded degree, it is shown that the compatibility condition is always satisfied.

\begin{lemma}[\citet{hutter2016optimal}]
Let $D$ be the incidence matrix of a graph $G$ with maximal degree $d$ and $\emptyset \ne T\subset E$. Then, $ \kappa_T \ge \frac{1}{2\min\{\sqrt{d},\sqrt{|T|}\}}$.
\label{prop:compat}
\end{lemma}

\begin{definition}[Inverse scaling factor]
The \it{inverse scaling factor} of an incidence matrix $D$ is defined as 
$
\rho(D): =\max_{j\in [m]} \|s_j\|$,
where $S=(D^T)^\dagger=[s_1^T,\cdots, s_m^T]^T$ is the pseudo inverse of $D^T$.
\label{def:inv_scaling_factor}
\end{definition}
By design of $D_v$, it is clear that $\rho(D_v)=\rho(D)$. 
Now we present the main result.
\begin{theorem}[Oracle Inequality for Projected TV Regression]
Under model \eqref{model} with $\epsilon_{i}\stackrel{i.i.d.}{\sim} \cN(0,\sigma^2I_M)$, define $\theta^*=\ttvec(\Gamma^{*T}X^T) $, $\that$ is the solution for \eqref{eq:cons_opt}. For any $\delta>0$, if $\lambda=\rho\sigma\sqrt{\log(mnM/\delta)}$, then with probability at least $1-\delta$,
\bas{
 \|\theta^*-\hat{\theta}\|_F^2 \le & \inf_{\bar{\theta}\in \bR^{n M}: H_v\bar{\theta}=\bar{\theta}} \left\{ \|\bar{\theta}-\theta^*\|^2+ 4\lambda \|(D_v^T \bar{\theta})_{T^c}\|_{\ell_1} \right\} \\
& \quad + 64\sigma^2\log\left( \frac{2enM}{\delta} \right)+ 8\rho^2\sigma^2\log\left(\frac{mnM}{\delta}\right)\kappa_T^{-2}|T|}
\label{th:consist}
\end{theorem}
The proof of Theorem~\ref{th:consist} is inspired by that in \cite{hutter2016optimal} and is deferred to Section~\ref{sec:proof_consistent} in supplementary materials.
\begin{remark}
The length of $\theta$ scales as $nM$. Hence we care about the mean recovery error $\frac{1}{nM}\| \theta^*-\hat{\theta} \|^2$, which converges at rate $O\left( \frac{\log(mnM)}{nM} \right)$. The upper bound exhibits the trade-off between two quantities: the number of non-smooth points $|T|$ and the total variation of the ``smooth" part  $\|(D\theta)_{T^c}\|_{\ell_1}$. To be more specific, given the total variation of the entire data fixed, when the model is piecewise smooth but have drastic change at the non-smooth points, $|T|$ will dominate $\|(D\theta)_{T^c}\|_{\ell_1}$; on the other hand, if the data has few non-smooth points but fluctuate a lot in each piece, the total variation of $\theta$ in $T^c$ will be large compared to $|T|$.
\end{remark}

If the graph has bounded degree, by Lemma~\ref{prop:compat} we immediately have the following corollary.
\begin{corollary}
If the maximal degree of the penalty graph $G$ is $d$, $\lambda=\rho\sigma\sqrt{\log(mnM/\delta)}$, then with probability at least $1-\delta$,
\bas{
\|\theta^*-\hat{\theta}\|_F^2 \le&  \inf_{\bar{\theta}\in \bR^{n M}: H_v\bar{\theta}=\bar{\theta}} \left\{ \|\bar{\theta}-\theta^*\|_F^2+ 4\lambda \|(D_v^T\bar{\theta})_{T^c}\|_{\ell_1} \right\} \\
&\qquad + 64\sigma^2\log\left( \frac{2enM}{\delta} \right)+ \frac{2\rho^2\sigma^2\log\left(\frac{mnM}{\delta}\right)|T|}{\min\{d,|T|\}}
}
\end{corollary}

The following Corollary F utilizes Theorem D to infer a bound on the parameter estimation.
\begin{corollary}
If  $\frac{1}{n}X^TX=I_p$, and $\lambda=\rho\sigma\sqrt{\log(mnM/\delta)}$, then with probability at least $1-\delta$,
\bas{
 \frac{1}{Mp}\|\hat{\Gamma}-\Gamma^*\|_F^2 \le &  \inf_{\bar{\Gamma}\in \bR^{p\times M}} \left( \frac{1}{Mp}\|\bar{\Gamma}-\hat{\Gamma}\|_F^2+ \frac{4\lambda}{ nMp} \|(X\bar{\Gamma}D)_{T^c}\|_{\ell_1} \right)  \\
&\quad + \frac{64\sigma^2\log\left( \frac{2enM}{\delta}\right) }{nMp }+ \frac{8\rho^2\sigma^2\log\left(\frac{mnM}{\delta}\right)|T| }{nMp\kappa_T^2}
}
\label{cor:gamma_oracle}
\end{corollary}

\begin{remark}
The condition in Corollary~\ref{cor:gamma_oracle} can be achieved by normalizing the input design matrix. When the covariance matrix for the features is not isotropic, the error term should be represented in Mahalanobis distance instead, that is, the error along different axis needs to be reweighted by the variance in that direction.  
\end{remark}

\section{Simulation Studies}
\label{sec:simu}

In this section, we present the simulation results on recovering signal $\Gamma$ with various methods. Since the goal is to accurately estimate $\Gamma$, we use the coefficient-mean deviation as the performance metric, which is defined as $\frac{1}{\sqrt{Mp}}\|\hat{\Gamma}-\Gamma^*\|_F$. The standard deviations of the metric over replicates are reported in the parenthesis.
We conduct experiments on 1D and 2D synthetic data.  

\subsection{One Dimensional Simulation}
\label{sec:1dsimu}
For 1D data, we mainly compared with the state-of-art Functional-on-Scalar regression methods (FoSR), it is known for good performance for 1D data, especially for data generated from its true basis assumption.  We generate two simulation scenarios, each with 200 replications. For each scenario, we present results for sample sizes of 25, 50, and 100.

In the first setting, we generate model that satisfies the assumptions of function-on-scalar regression, i.e., the signals are sparse if expanded in Fourier basis. The predictors are generated as following: $(X_1, X_2)$ are drawn from categorical distribution that takes value $(1,0)$ with probability $1/4$, $(0,1)$ with probability $1/4$, and $(0,0)$ with probability $1/2$. $X_3$ is drawn from standard normal distribution. The true signals are from Fourier basis functions, with index $t$ ranges from $1$ to $200$, and $$Y=0.3 \sin (\pi t/100)+0.5 X_1 \cos (\pi t/100)-0.3 X_2 \sin (\pi t/50)+ 0.5 \cos(\pi t/25) + 2 \mathcal{N}(0,1).$$ We compared our proposed method with FoSR using b-spline and Fourier basis functions. 

Table \ref{1Dsimu} presents the simulation results for scenario 1. 
When sample size is $25$, our proposed GFMR performs better than all the competitors including FoSR with the true Fourier basis, in terms of smaller mean and standard deviation. Our method performs similarly with the FoSR methods when sample size is $50$. With large sample size $100$, the FoSR methods has a better performance. The Fourier basis function performs slightly better then b-spline basis, however the difference is not phenomenal. FoSR is implemented in R package \texttt{refund} and we applied it with $L_2$ penalty, where the tuning parameter is selected through cross validation from a grid of $(2^{-5},2^{-3},2^{-1},2,2^3,2^5)$. 

\begin{table}[ht]
\centering
\footnotesize
\caption{Mean Deviation of  Coefficients for 1D settings}
\begin{tabular}{l|p{1cm}|llll}
  \hline
& sample size    & GFMR         &periodic\_GFMR            & FoSR\_bspline         & FoSR\_fourier     \\ 
  \hline
\multirow{3}{*}{Setting 1}      &25         &{\bf 0.045(0.008) }    & - &0.054(0.015)     & 0.053(0.015)          \\ 
                      &50         & 0.024(0.004)         & - &0.024(0.006)     & 0.024(0.006)          \\ 
                    &100     & 0.015(0.002)         & - &0.012(0.003)     &{\bf 0.011(0.003)}      \\ 
   \hline
 \multirow{3}{*}{Setting 2}    &25         & 0.076(0.009)         & {\bf 0.033(0.007)}     & 0.081(0.015) & 0.078(0.015)\\ 
                     &50         & 0.041(0.004)         &{\bf 0.020(0.004)}     & 0.051(0.006) & 0.049(0.006)  \\ 
                    &100     & 0.020(0.002)         & {\bf 0.012(0.002)}    & 0.039(0.003) & 0.037(0.004) \\ 
   \hline
\end{tabular}
\label{1Dsimu}
\end{table}
The second simulation setting showed the advantage of our methods when the FoSR models are misspecified. It also demonstrates the advantage of our proposed method in the flexibility for defining arbitrary adjacency structure with a self-defined incidence matrix.  
We generate piecewise constant signals, and in addition, the signals for indexes 1 to 100 are identical with signals in 101 to 200. This setting is motivated by the scenario in time series observations, where one may observe the same pattern over repeated time periods.  For example, one may observe repeated patterns in activity data or task fMRI data. 
We use this setting to demonstrate the performance of incorporating some prior knowledge for the repeated patterns is better than not using this knowledge. 

We consider two regimes in defining the edges. The GFMR only assumes edges connecting measurements in adjacent times: $E=\{(i,i+1), i=1,\dots,199\}$. The periodic\_GFMR adds 100 more edges that connect the pairs of measurements with a time lag of $100$, i.e., $(i, 100+i)$, which utilize the prior information that the signal is periodic.
Therefore, the edge set of periodic\_GFMR is $E=\{(i,i+1), i=1,\dots,199\} \cup \{(i,i+100),i=1,\dots,100\}$. 
The true signal is generated by the following model:
\bas{
&Y_{ij} = \mathbb{I}(1\le j\le 20)+ \mathbb{I}(101\le j \le 120)+ 0.5 X_{1i} \left(\mathbb{I} (31\le j \le 70)+\mathbb{I}(131\le j \le170)\right) \\
& -X_{2i} \left(\mathbb{I}(71 \le j \le 80)+\mathbb{I}(171 \le j \le 180)\right) + X_{3i} \left(\mathbb{I}(61 \le j \le 100)+\mathbb{I}(161 \le j \le 200)\right) + 2\mathcal{N}(0,1).
}
The results are shown in Table \ref{1Dsimu}. 
We can see the proposed GFMR, and periodic\_GFMR outperforms the other methods. As sample size increases, all the methods improves with smaller mean and SD for the mean deviation. By adding the additional edges, our proposed method shows better performance, and the advantage is more phenomenal with a small sample size. The proposed method, with added edges to incorporate prior knowledge and encourage similar patterns, reduced the mean deviation to less than half for the one estimated without added edges. 

\subsection{Two Dimensional Simulation}
\label{sec:2dsimu}
\begin{table}
\centering
\footnotesize
\caption{Mean Deviation of Coefficients for 2-D settings}
\vskip 0.1in
\begin{tabular}{c|p{1cm}|lllll}
\hline
        & sample size    & GFMR        & Envelope     & FoSR    & TV\_OLS    &OLS\_TV \\
        \hline
Setting 1    & 25            & {\bf 0.298 (0.038)}     & 0.349(0.028)        & 0.357 (0.019)         & 0.501 (0.044)         & 0.424 (0.048)\\
        & 50            & {\bf 0.217 (0.022)}     & 0.241(0.013)        &0.339 (0.010)            &0.361 (0.028)         &0.313(0.033)\\
        & 100        & 0.200 (0.016)        & {\bf 0.169(0.007)}    &0.330 (0.004)             &0.292 (0.017)         &0.230 (0.013)\\
        \hline
Setting 2    & 25            &{\bf 0.265 (0.008)}     & 0.561(0.045)        &0.321 (0.003)             & 0.275 (0.006)     & 0.919 (0.087)\\
        & 50            & {\bf 0.192 (0.007)}    & 0.387(0.019)        &0.318 (0.001)             & 0.266 (0.003)     & 0.545 (0.033)\\
        & 100        & {\bf 0.126 (0.006)}     & 0.267(0.010)        &0.318 (0.001)             &0.259 (0.003)        & 0.277 (0.012)
\\\hline
\end{tabular}
\label{table2}
\end{table}
For 2D data, 
We compared the proposed method GFMR with the following methods. We consider the set of edge from the grid graph. We provide a function (creategrid) to generate adjacency matrix for grid graph of 1-D to 3-D applications in the software package \url{https://github.com/summeryingliu/imagereg}. 
\begin{itemize}
\item \texttt{Envelope}: Parsimonious Tensor Response Regression \citep{li2017parsimonious}, implemented with matlab package provided by the authors.
\item  \texttt{FoSR}: A Variational Bayes implementation for penalized splines \citep{goldsmith2016assessing}, implemented using a Variational Bayes algorithm in R package \texttt{refund}.
\item \texttt{TV\_OLS} : A two-stage method: 1) apply TV denoising seperately for each image, 2) conduct a voxel-wise linear regression on the denoised images.
\item  \texttt{OLS\_TV}: A two-stage method: 1)Voxel-wise linear regression to get the noisy coefficient maps; 2) apply TV-denoising for each coefficient map.
\end{itemize}

\begin{figure}[t]
\small
\begin{tabular}{cc}
\hspace{-1em}\includegraphics[width=0.52\textwidth]{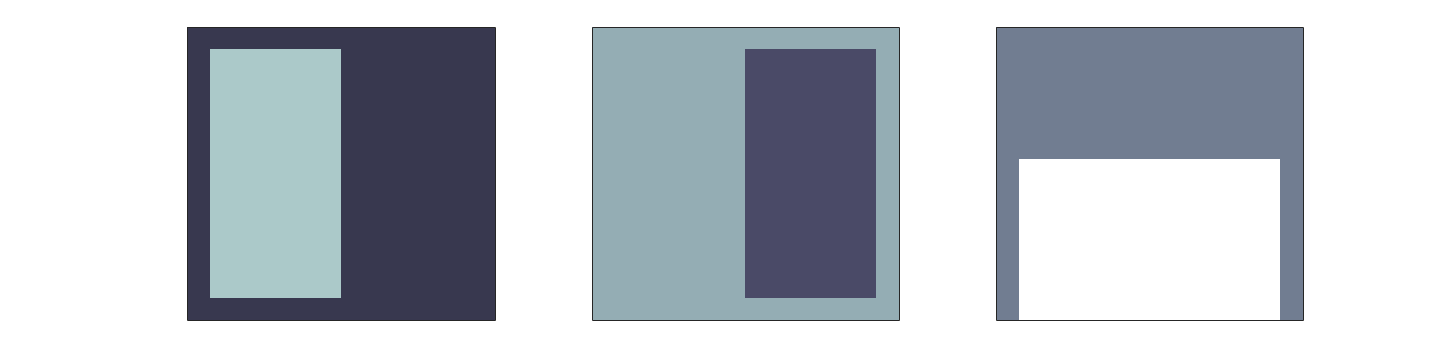} \hspace{-2em}         &\includegraphics[width=0.52\textwidth]{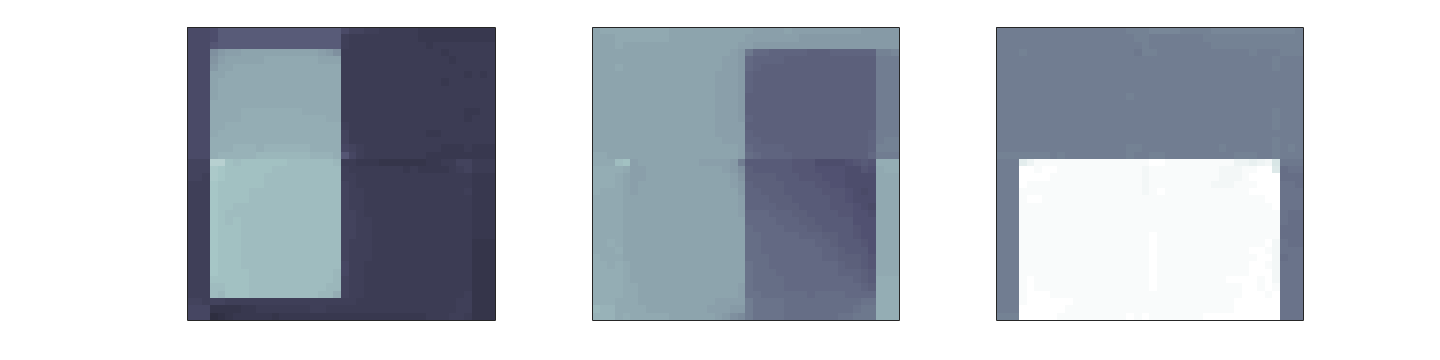}    \\
(a) Ground truth & (b) GFMR\\
\hspace{-1em}\includegraphics[width=0.52\textwidth]{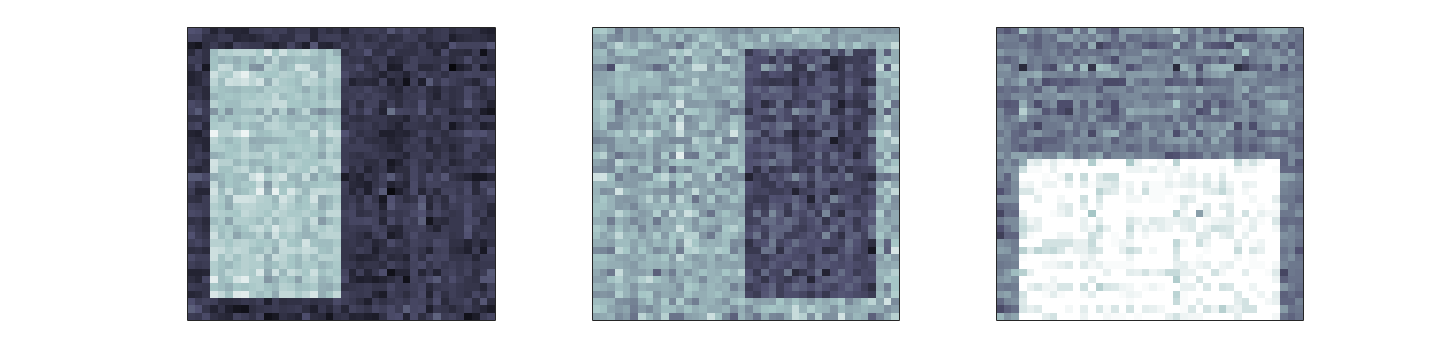} \hspace{-2em}         &\includegraphics[width=0.52\textwidth]{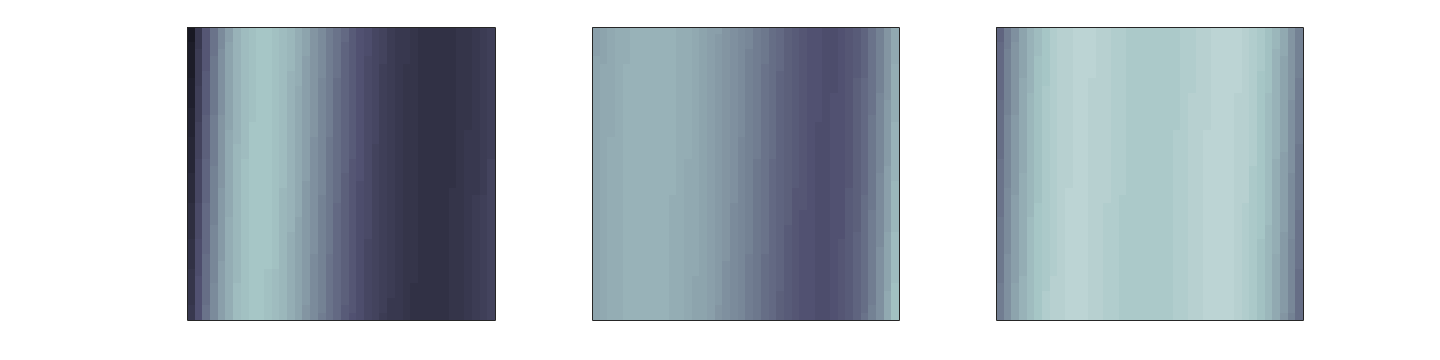}    \\
(c) Envelope & (d) FoSR\\
\hspace{-1em}\includegraphics[width=0.52\textwidth]{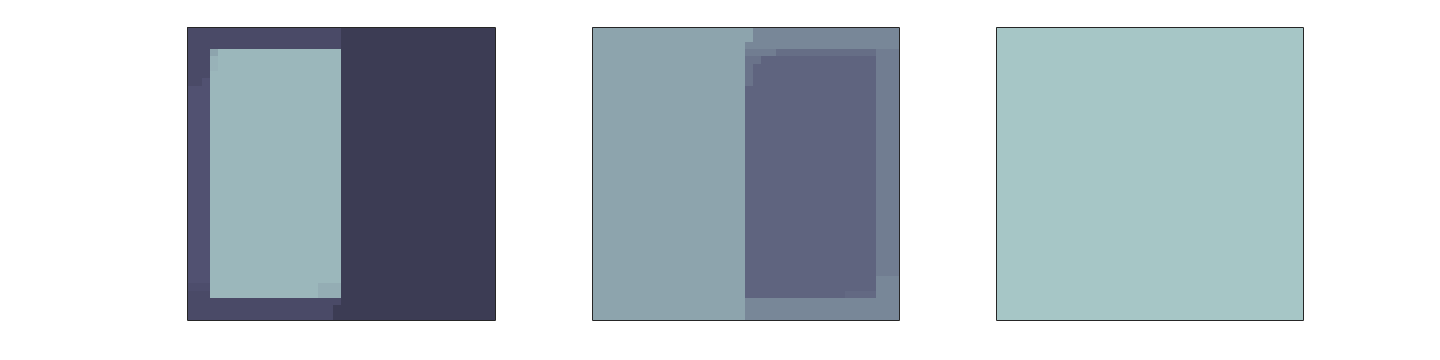} \hspace{-2em}     &\includegraphics[width=0.52\textwidth]{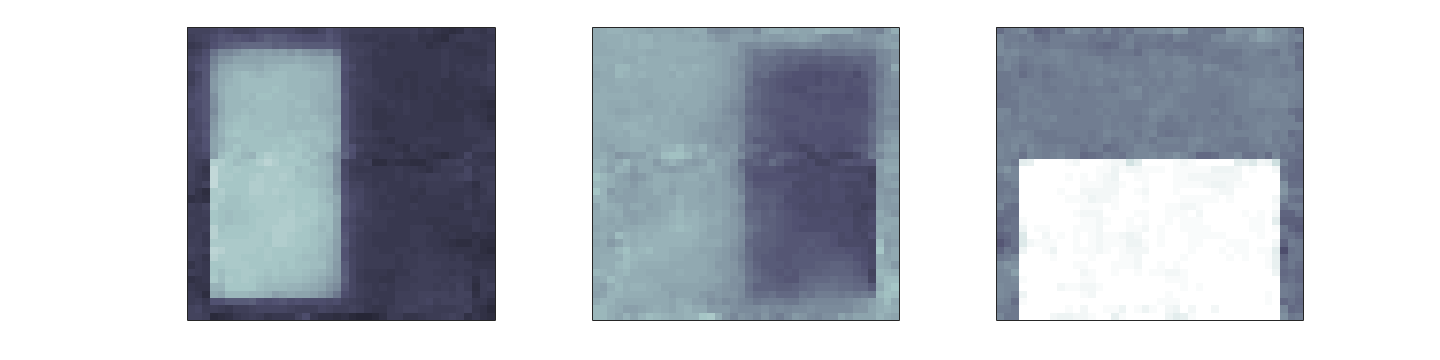}    \\
(e) Two step: OLS\_TV & (f) Two step: TV\_OLS
\end{tabular}
\vspace{1cm}
\caption{Coefficient Maps for Simulation Setting 1 in 2D.}
\label{coefexample1}
\end{figure}

\begin{figure}[t]
\small
\begin{tabular}{cc}
\includegraphics[width=0.44\textwidth]{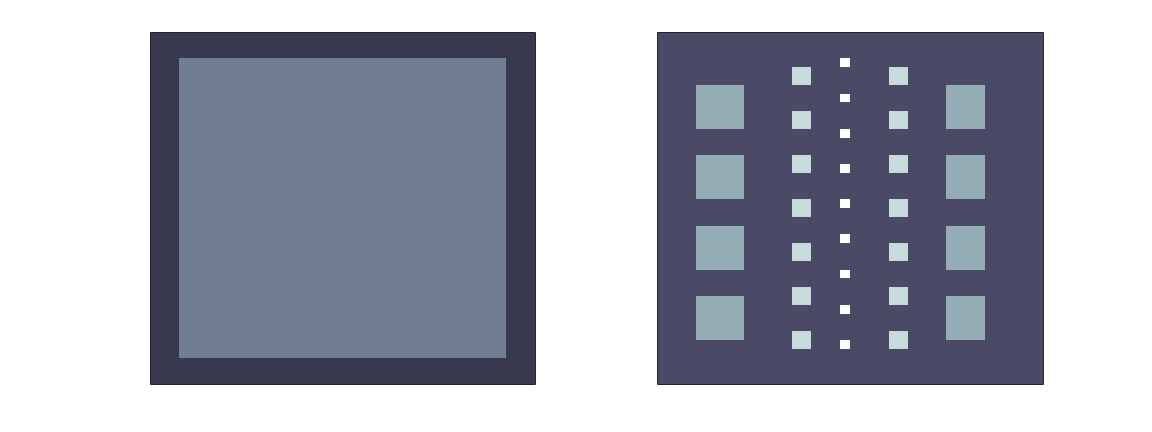}        &\includegraphics[width=0.44\textwidth]{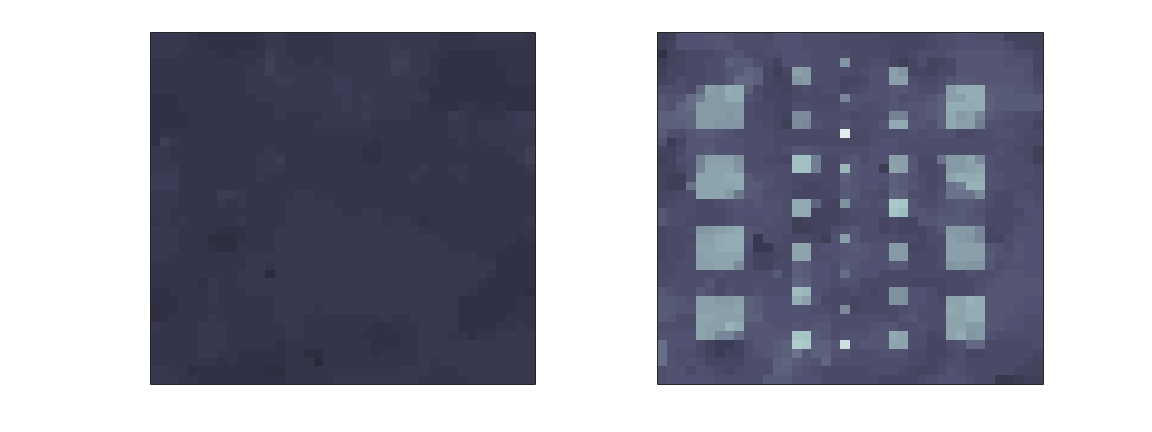}    \\
(a) Ground truth & (b) GFMR\\
\includegraphics[width=0.44\textwidth]{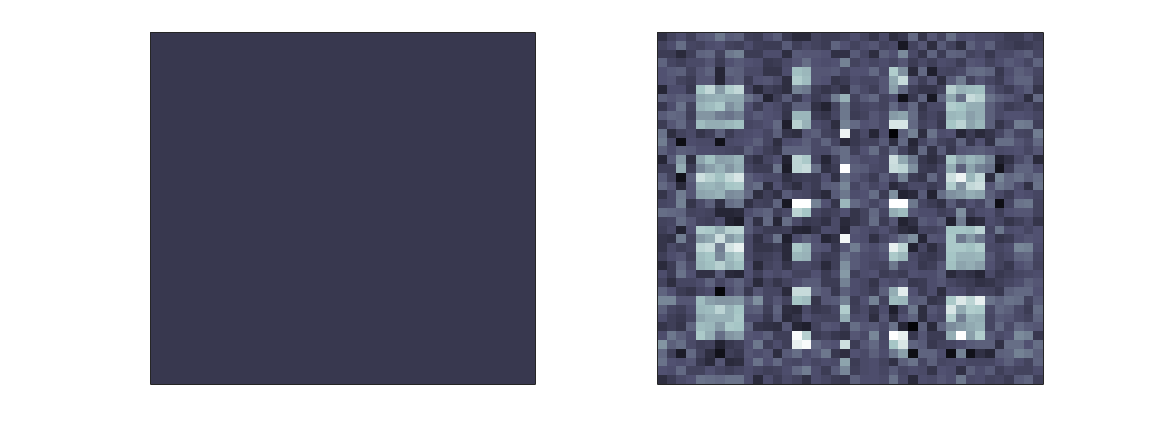}        &\includegraphics[width=0.44\textwidth]{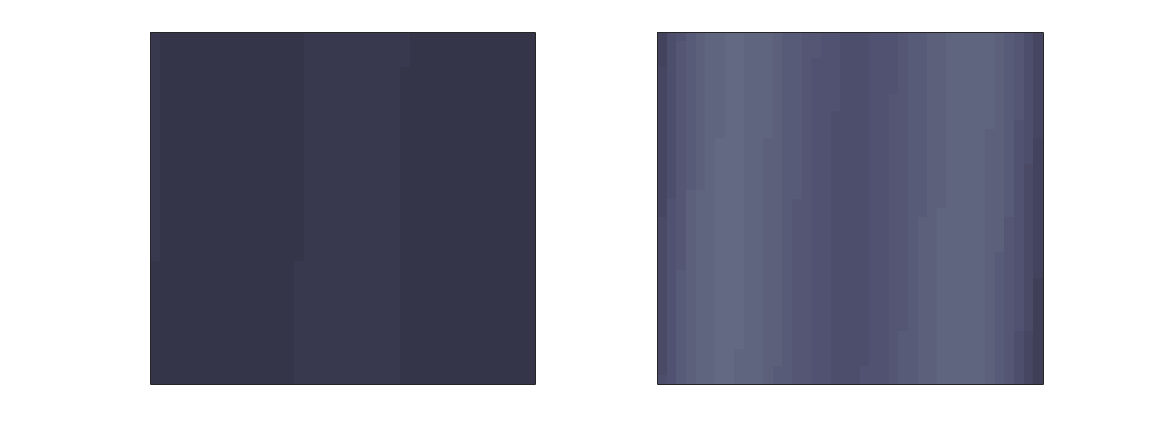}    \\
(c) Envelope & (d) FoSR\\
\includegraphics[width=0.44\textwidth]{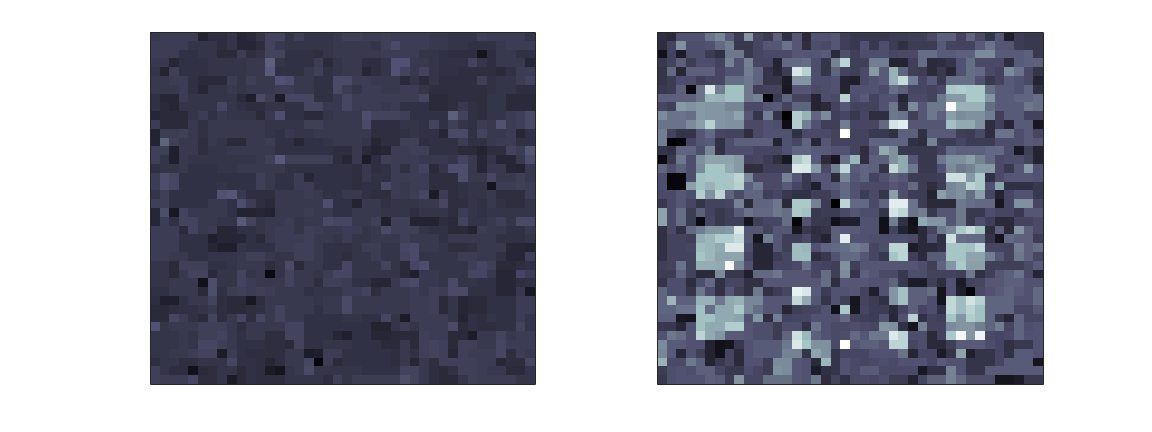}        &\includegraphics[width=0.44\textwidth]{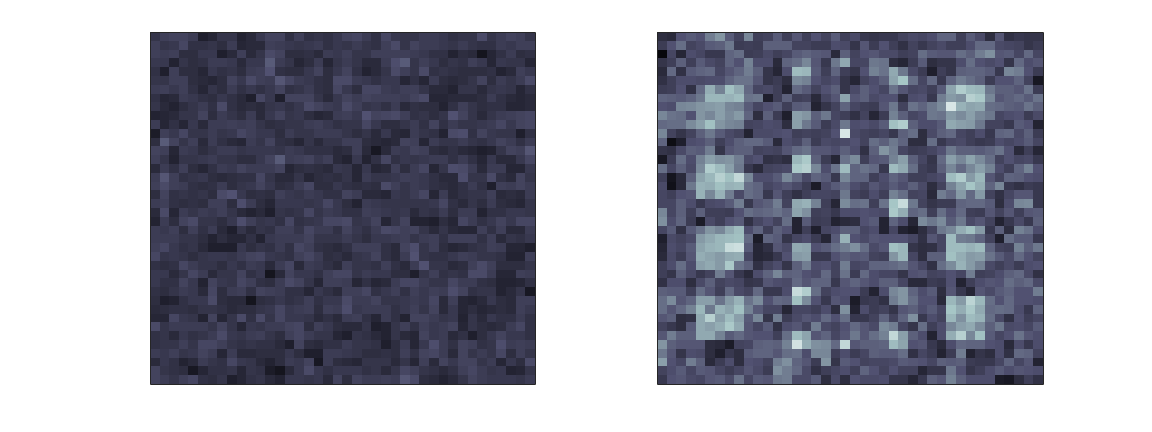}    \\
(e) Two step: OLS\_TV & (f) Two step: TV\_OLS
\end{tabular}
\vspace{1cm}
\caption{Coefficient Maps for Simulation Setting 2 in 2D.}
\label{coefexample2}
\end{figure}

We consider two different settings in the simulation. 
 For the first setting, the predictors are generated as the following: $(X_1, X_2)$ are drawn from categorical distribution that takes value $(1,0)$ with probability $1/4$, $(0,1)$ with probability $1/4$, and $(0,0)$ with probability $1/2$.  And an integer valued variable $X_3$ is generated uniformly over integers from $56$ to $75$ to represent the age of a patient.

The image size is $m_1=40$ by $m_2=40$,  the true signals are block-size constant with large block sizes, their coefficient maps ($\Gamma_1,\Gamma_2,\Gamma_3$) are shown in Figures~\ref{coefexample1} (a). And the outcome is generated by 
\bas{
Y = X_1\Gamma_1 + X_2\Gamma_2 + X_3\Gamma_3+\epsilon,
}
where $\epsilon\in \bR^{m_1\times m_2}$, and $\epsilon_{ij} \stackrel{i.i.d.}{\sim} \cN(0,2)$.

In the second setting, two binary variables $X_1$ and $X_2$ are generated the same as in setting 1. This setting has various block sizes, the coefficient map of $X_1$ is generated to have active regions with different sizes: $1$, $4$,  $25$ pixels, and the true coefficients are $2$, $1.5$ and $1$, respectively. Please refer to Figure \ref{coefexample2} (a) for the actual arrangements of these true signals. 


Table~\ref{table2} demonstrates the clear advantage of our proposed methods in recovering the coefficient maps for both settings in terms of a smaller mean deviation of the estimated coefficient from the truth. With increasing sample sizes, all methods perform better. When the sample size is 100, the envelope method has a better performance in setting one which large block size for all signals.  The envelope method is outperformed by our proposed method and  \texttt{TV\_OLS} (regression after TV denoising) for setting 2 with all three sample sizes when the block has various sizes. 

We further illustrate the estimated coefficient map from various methods in Figure~\ref{coefexample1} and~\ref{coefexample2} with one example in sample size 25.  Our proposed method achieves a cleaner cut on the boundary and smoother inside the signal blocks. The envelope method is able to capture and display the main pattern, however, does not encourage sparsity of the edges. Thus the solution is not ``smooth" anywhere. 
There is no public available package/code to conduct 2D FoSR as far as we know, so implemented it with input as vectorization of the 2D outcome, the performance is thus not ideal.

Our proposed method is better than the two-step methods in all settings because it is simultaneously conducting regression and total variation smoothing. Thus it can utilize all samples to capture signals with various block sizes. In comparison, if one conducts 'smoothing' before regression (TV\_OLS), the signal in small blocks is more likely to be 'smoothed' away in the image level. This intuitive explanation is consistent with the simulation result that the mean deviation of TV\_OLS does not improve much when sample size increases in setting 2, which contains various block sizes.  Furthermore, the estimated coefficient maps of TV\_OLS are more 'vague' at the boundaries, since the smoothing step is done separately for individual images, the edges are not overlapping across the smoothed images. 
On the other hand, if one conducts 'smoothing' after regression (OLS \_TV), the random noise of voxel-wise coefficients and signals of small effect size will be 'smoothed' in the same step, thus hard to differentiate. With the increasing sample size, the variability of coefficients across pixels will be smaller. However, our simulation shows our proposed method performs much better than OLS\_TV in all settings and the presented sample sizes.

\section{Physical Activity Data}
\label{sec:real}
The first example came from the accelerometry subsample of a large community-based study of mood disorders, the National Institute of Mental Health Family Study of Affective Spectrum Disorder \citep{Merikangas2014}. The study collected minute-level physical activity measures (i.e., activity counts) from about 350 participants using the Philips Respironics Actiwatch for continuous 24 hours over 2 weeks. The outcome measures, time series of activity counts, are densely sampled ($>$ 20,000 values per participant) and highly noisy. Hence we aim to enhance our power in assessing the association between physical activity patterns with covariates of interest such as demographic characteristics and health-related factors by simultaneously denoising the data with regression modeling. 

The activity counts were transformed as log (counts+1) given the skewness \citep{Shou2017}, and pre-smoothed with a $60$-minute sliding window. 
Under the assumption that physical activity status (e.g., being inactive, moderately or vigorously active) has temporal continuity despite the noisy level shown in the raw data (Figure \ref{raw}), we believe that such data could benefit from GFMR. The covariates include diagnosis, age, gender, body mass index (BMI), and day of the week. In particular, the diagnosis was determined by DSM-IV and contained 5 exclusive groups: healthy control, type I bipolar (BPI), type II bipolar (BPII), major depressive disorder (MDD), and other disorders. 
To make the results more clinically interpretable, we categorized BMI into underweight ($<18.5$), normal ($18.5\le$BMI$\le25$), overweight ($25<$BMI$\le30$) and obese (BMI$>30$), with the normal participants being the reference group. Similarly, age was also stratified into four categories: adolescence (under 18), adulthood (18 to 40), middle-age (40 to 60), and elderly ($\ge 60$). In total, there were 4214 daily activity curves, and the observations were balanced for the 7 days of the week (Monday through Sunday). 

\begin{figure}[h!]
 \centering
\begin{tabular}{cc}
\includegraphics[width=0.4\textwidth]{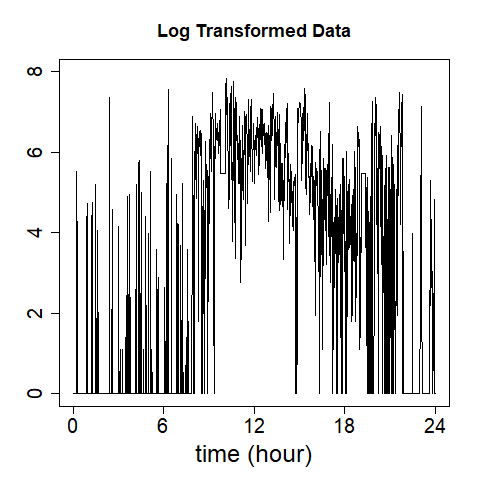}\hspace{-.6em}&
\includegraphics[width=0.4\textwidth]{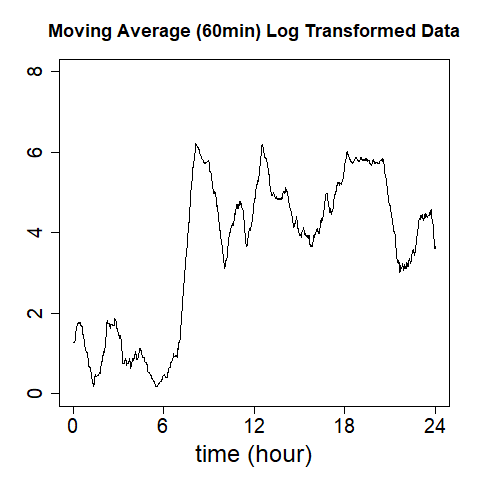} \hspace{-.6em}
\end{tabular}
\vspace{1cm}
\caption{An example of the observed daily activity counts after log transformation with (right) and without (left) pre-smoothing with moving average. }
\label{raw}
\end{figure}

We consider a edge set of all pairs of adjacent time stamps. The estimated time-dependent regression coefficients are shown as red curves in Figure \ref{fig:asir1}. The corresponding point-wise confidence bands were obtained as the upper and lower $2.5\%$ quantiles over $100$ bootstrap samples. Compared with the middle-age group, seniors tend to have an overall lower intensity throughout the day, especially from late afternoon to evening. Both adolescent and adult groups have shown higher night-time activity and lower early morning activity, but the difference is more prominent between adolescent and middle-age group. Participants who were obese or overweight were observed to have lower day-time activity and higher night-time activity as compared to those with normal BMI, but the overweight group had insignificant results. The underweight subjects, compared to normal, had a reversed trend. In terms of diagnostic groups, BPI subjects showed lower activity intensity later of the day (after 12 pm) as compared with healthy controls, which is consistent with previous findings \citep{Shou2017}. Females were also observed to have higher average activity intensities during day time and lower intensities at night. Interestingly, the effects of the day of the week were prominent in our results. With Sunday as the reference, we independently observed similar patterns across weekdays (Monday to Friday), where the morning activity intensities are much higher than Sunday, but no difference is shown for the rest of the day. Note that on Friday, we observed a rising trend of activity towards midnight, indicating that subjects have later bedtime on Friday. Saturday has a non-differentiable pattern compared with Sunday, but a similar rising trend towards midnight.

\begin{figure}[H]
\centering
\vspace{-2.7em}
\caption{ Activity Data Analysis Results. } 
\label{fig:activity}
\begin{tabular}{cccc}
\hspace{-2em}
\includegraphics[width=0.25\textwidth]{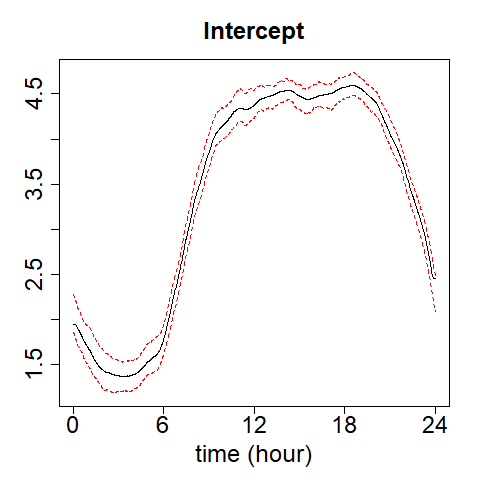}\hspace{-2em}&
\includegraphics[width=0.25\textwidth]{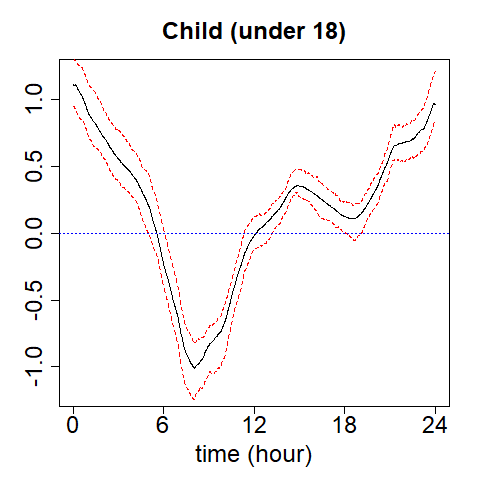} \hspace{-1em}&
\includegraphics[width=0.25\textwidth]{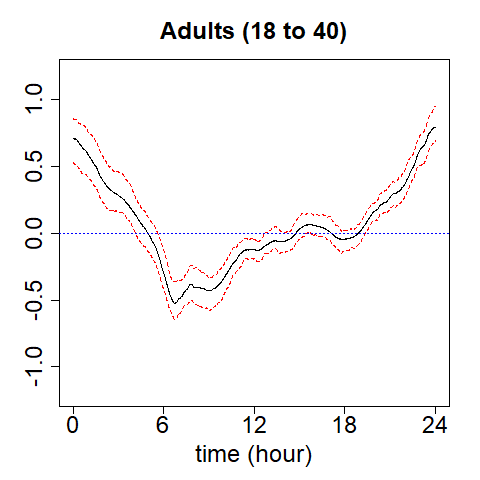} \hspace{-1em}&
\includegraphics[width=0.25\textwidth]{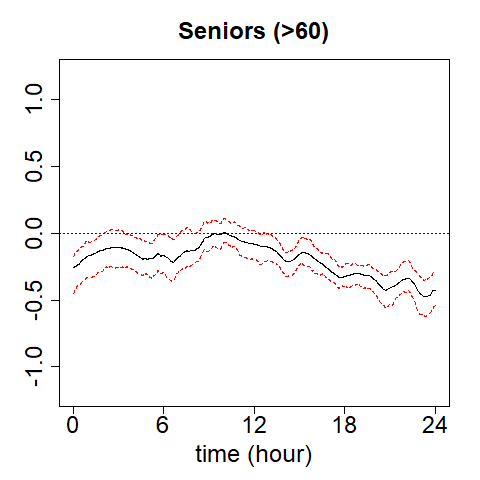} \hspace{-1em}\\
\hspace{-2em}
\vspace{-0.7em}
\includegraphics[width=0.25\textwidth]{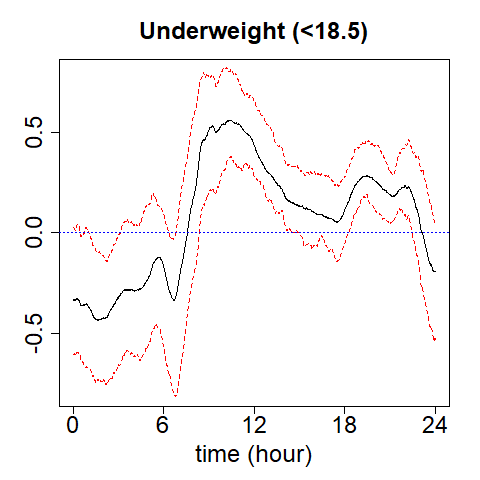}  \hspace{-1em}&
\includegraphics[width=0.25\textwidth]{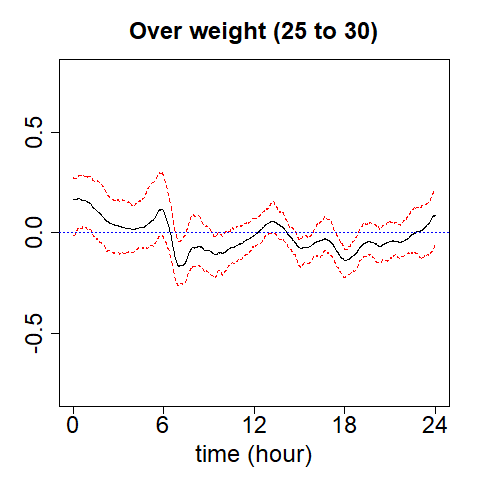} \hspace{-1em}&
\includegraphics[width=0.25\textwidth]{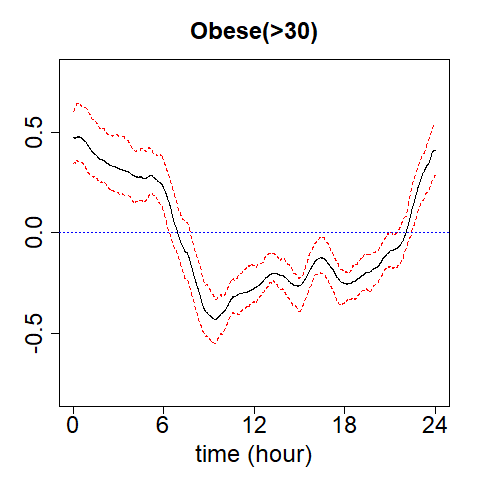}\hspace{-1em}&
\includegraphics[width=0.25\textwidth]{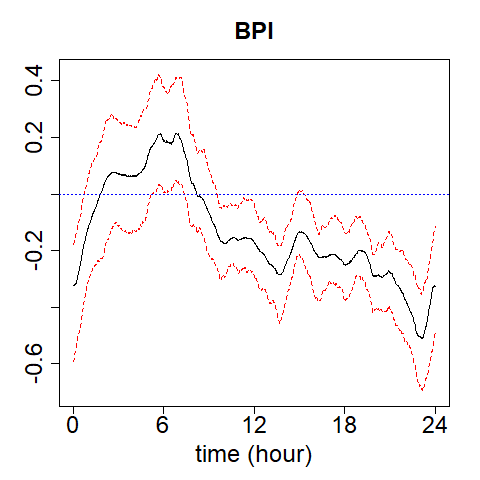}\hspace{-1em}\\
\hspace{-2em}\vspace{-0.7em}
\includegraphics[width=0.25\textwidth]{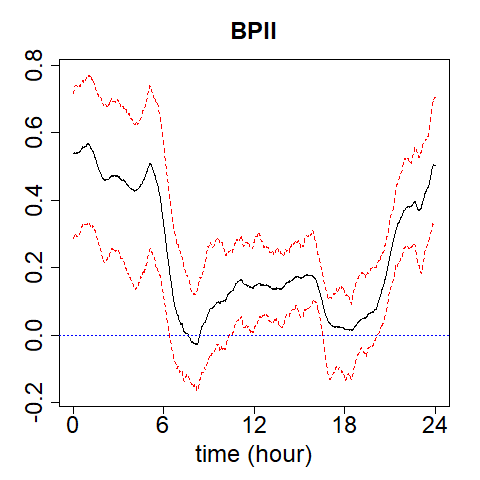}\hspace{-1em}&
\includegraphics[width=0.25\textwidth]{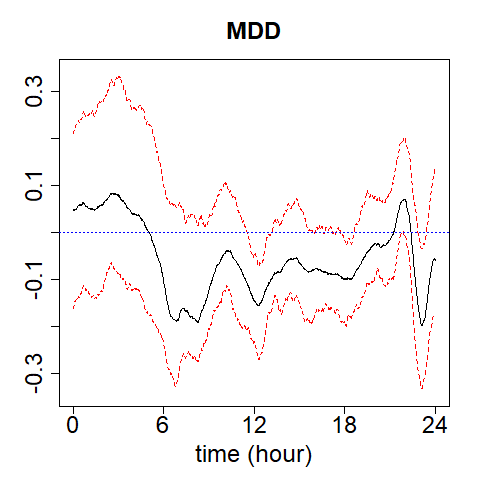}&
\includegraphics[width=0.25\textwidth]{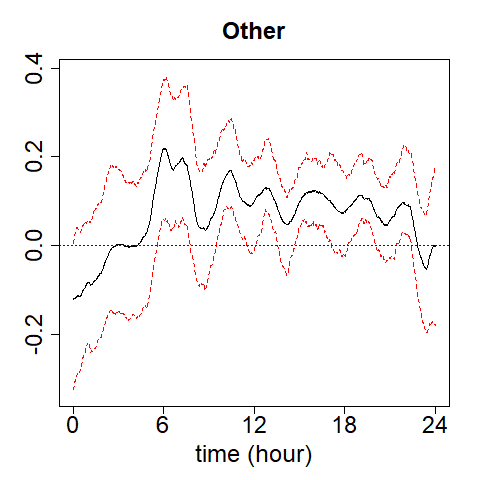}&
\includegraphics[width=0.25\textwidth]{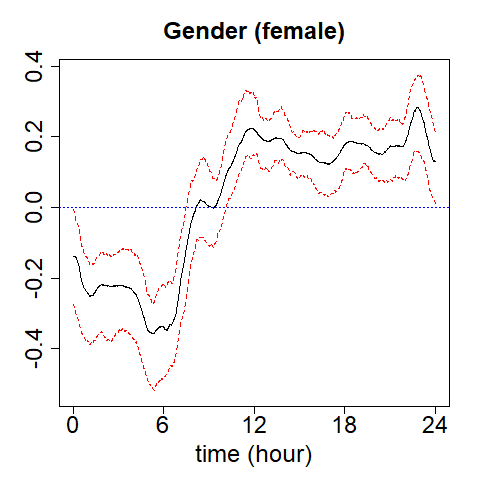}\\
\hspace{-2em}\vspace{-0.7em}
\includegraphics[width=0.25\textwidth]{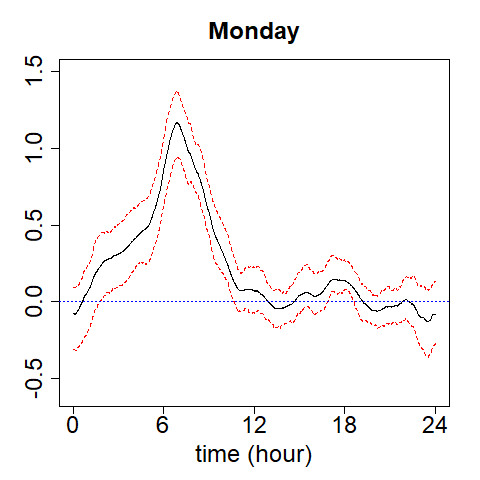}\hspace{-1em}&
\includegraphics[width=0.25\textwidth]{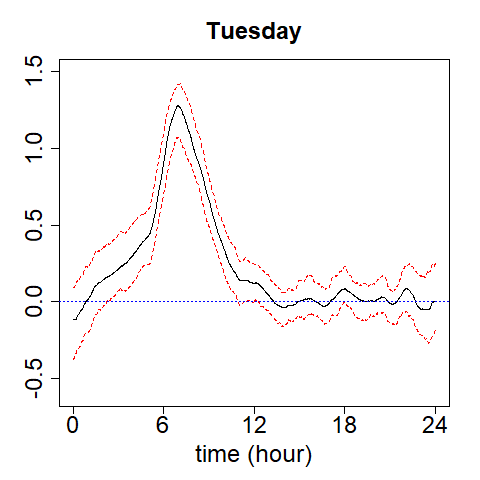}\hspace{-1em}&
\includegraphics[width=0.25\textwidth]{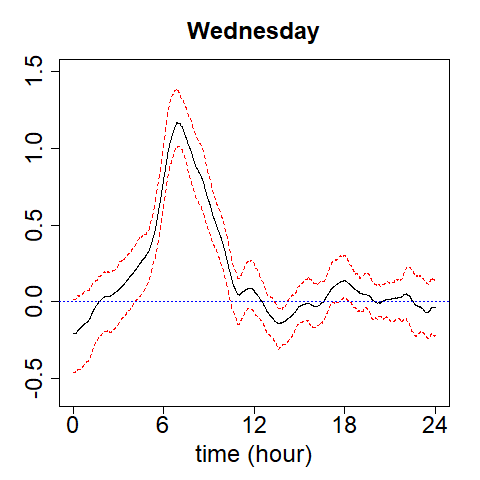}\hspace{-1em}&
\includegraphics[width=0.25\textwidth]{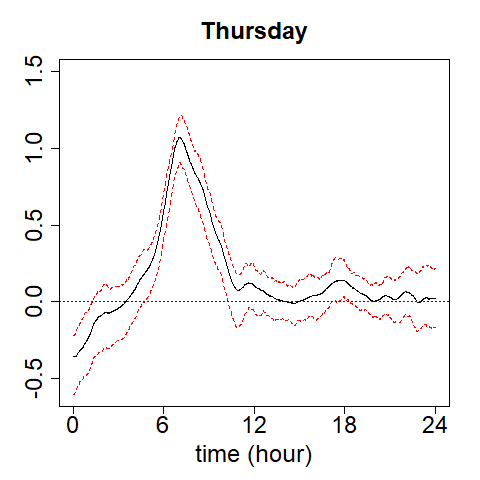}\hspace{-1em}\\
\hspace{-2em}\vspace{-0.7em}
\includegraphics[width=0.25\textwidth]{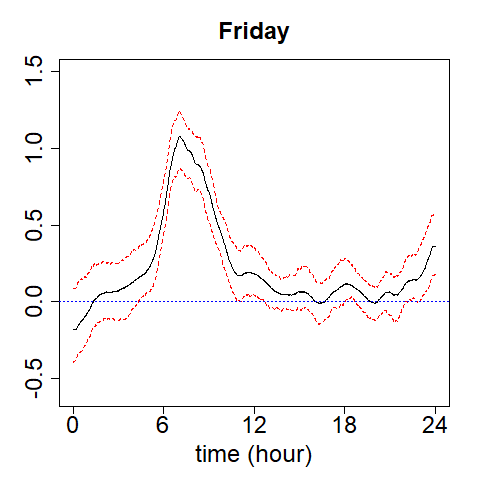}\hspace{-1em}&
\includegraphics[width=0.25\textwidth]{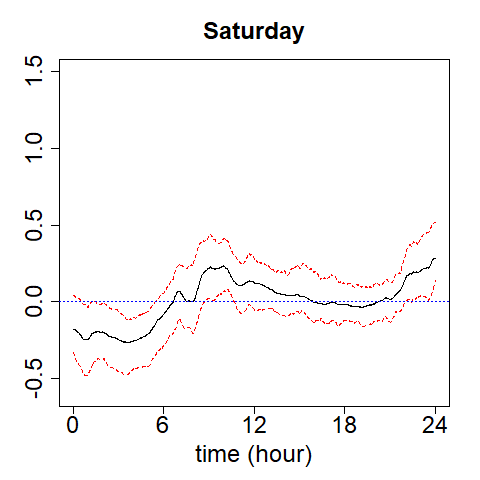}\hspace{-1em}\\
\end{tabular}
\label{fig:asir1}
\end{figure}

\section{3D Brain Imaging Analysis}
We now demonstrate the proposed method for analyzing 3D brain imaging. We are again using the grid graph for adjacency matrix. Our data came from a public data source of attention hyperactivity disorder (ADHD), the ADHD-200 sample, which was an initial effort to understand the neuro basis of the disease by ADHD consortium. The dataset contains structural and BOLD functional magnetic resonance imaging (MRI) scans from 776 subjects, including 491 typically developing controls (TDs) and 285 ADHD children.

We focused on the gray matter voxel-based morphometry (VBM) maps preprocessed via the Burner pipeline using SPM8, downloaded from the Neuro Bureau (\url{https://www.nitrc.org/plugins/mwiki/index.php?title=neurobureau:BurnerPipeline}). The intensities of the images quantify the normalized gray matter concentration changes at specific brain locations. The same dataset was used in \cite{li2017parsimonious}, where the performance of OLS and tensor envelope based regression were compared. 
Following \cite{li2017parsimonious}, we downsized the images to 30$\times$ 36$\times$ 30 using TensorReg from the original dimension of 121$\times$ 145$\times$ 121. The covariates in the model include ADHD, age, gender, and handedness defined as 1 for left-handed children and 0 for right-handed. We also include the site indicators in the multivariate analysis to adjust for the systematic site effects. The estimated regression coefficients were mapped back to the 3D population-average template of the Burner VBM images. The Harvard-Oxford Atlas \citep{Bohland2009} were used for references of the cortical regions. 

Figure \ref{fig:ADHDcoef} shows the brain regions with the highest regression coefficients based on GFMR, which corresponds to areas where the gray matter concentrations are mostly associated with ADHD, age, gender, and handedness. The areas where the coefficients' magnitudes are in the 10\%, 5\%, 2.5\%, and 0.5\% tails across the whole brain were displayed. 
We observed that after adjusting for age, gender, handedness, and site-effects, ADHD showed the most substantial reduction in the gray matter volume in the frontal lobe, such as the anterior and posterior cingulate gyrus and frontal operculum. Such regions have been known to be associated with sensory and motor responses, as well as cognitive and emotional functions. Significant cortical thinning in anterior cingulate gyrus has been reported among children with ADHD \citep{vanRooij2015}. These findings of the frontal lobe were novel in our analysis and were not captured in \cite{li2017parsimonious}. Instead, their methods identified superior temporal gyrus, pyramid, and uvula in the cerebellum. Additionally, we observed that cortical thickness increases in the frontal operculum with older children, which is linked with better task controls\citep{Higo2011}. Females are shown to have a thicker temporal fusiform cortex and frontal medial cortex. Interestingly, although the magnitude of the handedness effects is small, we observed an asymmetric association where the left-handed children tend to have greater gray matter thickness on the left hemisphere of the brain in regions such as posterior parahippocampal gyrus, lingual gyrus, and temporal occipital fusiform cortex. Such findings are consistent with several previous studies \citep{cuzzocreo2009}.

The 3D outcome is more computationally challenging since the number of edges triples from 1-D to 3-D. This specific example consists of images from 770 subjects; each image consists of 32400 nodes and 94140 edges for a grid smooth structure.  Algorithm \ref{alg:palm} offers an iterative procedure of solving a TV denoising problem. In this real data example and our simulation studies, the algorithm converges in 10-20 iterations, when the parameters are initialized with random noise following a normal distribution. The most computationally expensive step of Algorithm \ref{alg:palm} is line 6, where one is solving a TV denoising problem for a graph with $32400\times 770$ nodes and $94140 \times 770$ edges. Solving TV denoising in a single step with this scale is very costly in memory. Our algorithm can solve this step distributively. Since there are no edges across subjects, one can distribute the fused lasso problem to multiple cores. i.e., within each iteration, at line 6,  one could divide the samples to several subgroups, solve for the $\mu^{(k+1)}$ for each subgroup, and combine the results together. The combined estimates of  $\mu^{(k+1)}$ are equivalent to the solutions from computing the combined TV denoising problem. For this dataset, we distributed the problem in lines 6 to 16 cores in a workstation, in batches with 8 samples and tuning parameter $\lambda=0.05$. 

\begin{figure}[h!]
 \centering
\begin{subfigure}{0.45\textwidth}
\includegraphics[height=1.4in]{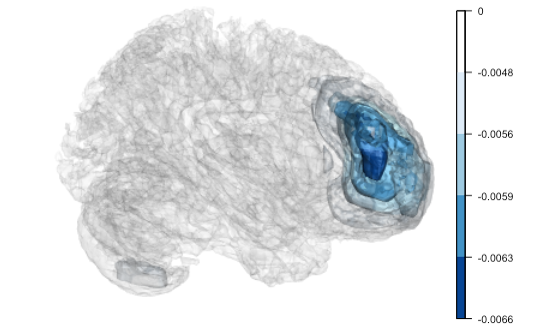}
\caption{ADHD}
\end{subfigure}
\begin{subfigure}{0.45\textwidth}
\includegraphics[height=1.4in]{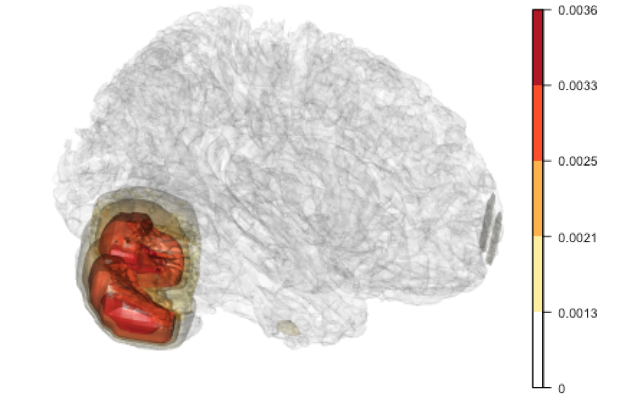}
\caption{Age}
\end{subfigure}
\begin{subfigure}{0.45\textwidth}
\hspace{.3em}
\includegraphics[height=1.45in]{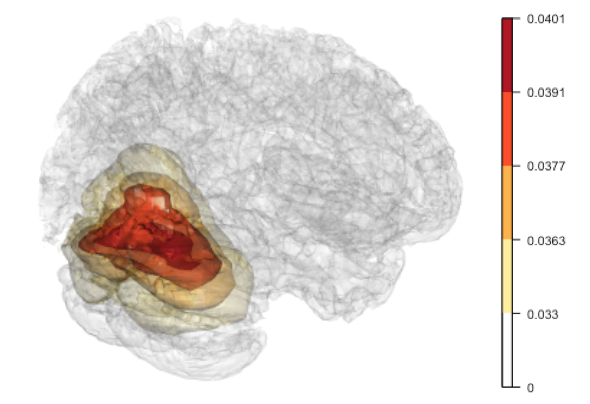}
\caption{Gender}
\end{subfigure}
\begin{subfigure}{0.45\textwidth}
\hspace{1em}
\includegraphics[height=1.5in]{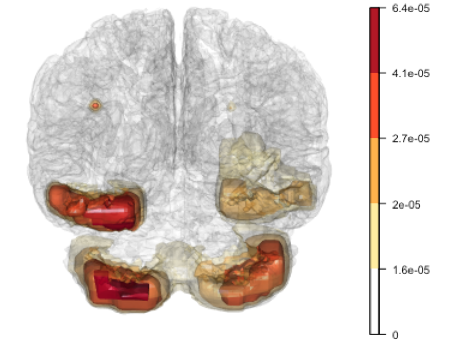} 
\caption{Left Handed}
\end{subfigure}

\caption{\small{ADHD analysis results: from light to dark are regions with estimated regression coefficients that lie above the 10\%, 5\%, 2.5\% and 0.5\% tails across the whole brain, receptively for covariates including ADHD, age, gender (females) and left-handedness. The color bars correspond to the positive or negative coefficients that are shown in the images.}}
 \label{fig:ADHDcoef}
\end{figure}

\section{Conclusion}
\label{sec:discussion}
We propose the Graph-Fused Multivariate Regression (GFMR) to solve the challenging regression problem with the noisy and high dimensional outcome. It uses a TV regularization term to encourage smoothness under flexible user-defined spatial and temporal structures. In addition, we propose an efficient and scalable algorithm that converges to the global solution, which has the desirable oracle property.
Simulations have shown that our method can adaptively achieve superior performance over various outcome dimensions and data generation mechanisms, including non-uniform sized signal. The advantage of GFMR is especially prominent when the sample size is relatively small. Our method is shown to be directly applicable to 1D accelerometry data and 3D brain imaging outcomes.  
Since GFMR allows to incorporate a user-defined adjacency matrix for smoothness, the method can be flexibly extended to data of higher dimensions such as 4D functional MRI. It also provides opportunities to integrate prior knowledge from other data modality or view. For example, when the outcome is structural imaging, the brain connectivity graph from fMRI could be input as an adjacency matrix to guide the TV penalty.

In this paper, we presented results for the adjacency matrix $D$ corresponding to the difference in means between the nodes connected by edges. Intuitively, it is most suitable for block-wise constant signals. On the one hand, the simple function approximation Lemma justifies our method is suitable for more general signal patterns since the block-wise constant signal can approximate any bounded function well, which explains the good performance of our proposed method in the 1D simulation setting 1, where the signals are not piecewise constant. One the other hand, our algorithm can be easily extended to incorporate other continuous constantly changing signals by extending the $D$ matrix to incorporate higher-order differences corresponding to 1st, 2nd or higher-order differentiation. The higher-order $D$ matrix can be implemented by the transformation proposed in \cite{steidl2006splines}; future work can examine whether it could improve the accuracy in some specific applications. Similar to the other penalized regression methods, our proposed method focuses on providing point estimation.  In translational medical research, in addition to signal recovering accuracy, proper inference of variability of the estimation is desirable. One can provide a confidence interval for the estimated coefficients through bootstraps. However, as the number of edges increases tremendously in the 3D/4D brain imaging applications, the computational burden became very heavy. Although our algorithm can be distributed to multiple cores to save time, computing bootstrap CI costs thousands of core hours and terabytes of memory on a high performance supercomputer.  It remains future work in providing a way to make the inference more feasible in 3D/4D applications.

\section*{Acknowlegements}
\label{sec:acknowledgement}

The authors would like to thank Drs. Xin Zhang and Lexin Li for providing codes and guidance of their envelope tensor regression algorithms. We also would like to acknowledge Junchi Li and Michael Daniels for discussion and advice to improve the quality of the work, and Jinshan Zeng for his comments on the convergence properties of the algorithm. HS's research was partially supported by the Intramural Research Program of the NIMH. The NIMH family study was conducted under the support of the Intramural Research Program (Merikangas; grant number Z-01-MH002804) and under clinical protocol 03-M-0211 (NCT00071786).

\section*{Disclaimer}
The views and opinions expressed in this article are those of the authors and should not be construed to represent the views of any of the sponsoring organizations, agencies, or US Government.


\section*{Data Availability Statement}
The code of our method that generate the numerical results were publicly available on \url{https://github.com/summeryingliu/imagereg}.
Real data example 1 (physical activity data) is research data not shared.
Real data example 2 (3-D brain imaging)  were derived from the following resources available in the public domain:\url{https://www.nitrc.org/plugins/mwiki/index.php?title=neurobureau:BurnerPipeline}

\appendix
\section{Proofs in Section~\ref{sec:model}}
\label{app:proof}
\subsection{Proof of Lemma~\ref{lem:theta2gamma}}
\begin{proof}[Proof of Lemma~\ref{lem:theta2gamma}]
Recall that $H_X$ is the projection matrix onto $\text{span}(X)^\perp$, therefore $\theta\in \text{span}(X)$ is equivalent to $H_X\theta=0$. By a variable transformation $\theta=X\Gamma$, solving (P)  is equivalent to solving (CP), and $\that=X\hat{\Gamma}$. When $X^TX$ is invertible, we solve $\hat{\Gamma}=(X^TX)^{-1}X^T \that$.
\end{proof}

\subsection{Proof of Theorem~\ref{thm:admm_converge}}
\begin{proof}[Proof of Theorem~\ref{thm:admm_converge}]
We rewrite problem~\eqref{eq:augmented_admm} as
\bas{
\min & \quad f(\theta)+g_1(\eta)+g_2(\mu)\\
s.t. & \quad A_1\theta+A_2\eta+A_3\mu=0
}
where
\bas{
&f(\theta)=\frac{1}{2}\|y-\theta\|^2, g_1(\eta) = \delta((I-H_v)\eta=0), g_2(\mu)=\|D_v^T\mu\|_1\\
& A_1 = \begin{bmatrix} -I\\-I\end{bmatrix}, A_2  = \begin{bmatrix} I\\0\end{bmatrix}, A_3  = \begin{bmatrix} 0 \\ I \end{bmatrix} \in \bR^{2nM\times nM}
}
Algorithm~\ref{alg:palm} is presented in the scaled form of ADMM. As is shown in~\citet{boyd2011distributed}, it is equivalent to the unscaled form as below.
Denote the Lagrangian as 
\bas{
L(\theta,\eta,\mu,U,V,\rho) = f(\theta) + g_1(\eta) + g_2(\mu) + \begin{bmatrix}U\\V\end{bmatrix}^T(A_1\theta+A_2\eta+A_3\mu)+\frac{\rho}{2}\|A_1\theta+A_2\eta+A_3\mu\|^2.
}
The unscaled form is
\bas{
& \theta^{(k+1)} = \arg\min_\theta  f(\theta) + \begin{bmatrix}U\\V\end{bmatrix}^T(A_1\theta)+\frac{\rho}{2}\|A_1\theta+A_2\eta^{(k)}+A_3\mu^{(k)}\|^2\\
& \eta^{(k+1)} = \arg\min_{\eta} g_1(\eta) + \begin{bmatrix}U\\V\end{bmatrix}^T(A_2\eta)+\frac{\rho}{2}\|A_1\theta^{(k+1)}+A_2\eta+A_3\mu^{(k)}\|^2\\
& \mu^{(k+1)} = \arg\min_{\mu} g_2(\mu) + \begin{bmatrix}U\\V\end{bmatrix}^T(A_3\mu)+\frac{\rho}{2}\|A_1\theta^{(k+1)}+A_2\eta^{(k+1)}+A_3\mu\|^2\\
& \begin{bmatrix}U^{(k+1)}\\V^{(k+1)}\end{bmatrix} = \begin{bmatrix}U^{(k)} \\V^{(k)} \end{bmatrix}+\rho(A_1\theta^{(k+1)}+A_2\eta^{(k+1)}+A_3\mu)
}
We first convert the three-block problem into a two-block problem with the technique in~\cite{chen2016direct}.

By first-order optimality condition in iteration $k$, we have
\bas{
& f(\theta)-f(\theta^{(k+1)}) + (\theta-\theta^{(k+1)})\left( -A_1^T(\lambda^k-\rho(A_1\theta+A_2\eta+A_3\mu)) \right)\ge 0, \forall \theta\in \bR^{nM}\\
& g_1(\eta) - g_1(\eta^{(k+1)}) + (\eta-\eta^{(k+1)})\left( -A_2^T(\lambda^k-\rho(A_1\theta+A_2\eta+A_3\mu)) \right)\ge 0, \forall \eta \in \bR^{nM}\\
& g_2(\mu)- g_2(\mu^{(k+1)}) + (\mu - \mu^{(k+1)})\left( -A_3^T(\lambda^k-\rho(A_1\theta+A_2\eta+A_3\mu)) \right)\ge 0, \forall \mu \in \bR^{nM}\\
}
Note that $A_2^TA_3 = 0$, we have
\bas{
& f(\theta)-f(\theta^{(k+1)}) + (\theta-\theta^{(k+1)})\left( -A_1^T(\lambda^k-\rho(A_1\theta+A_2\eta+A_3\mu)) \right)\ge 0, \forall \theta\in \bR^{nM}\\
& g_1(\eta) - g_1(\eta^{(k+1)}) + (\eta-\eta^{(k+1)})\left( -A_2^T(\lambda^k-\rho(A_1\theta+A_2\eta)) \right)\ge 0, \forall \eta \in \bR^{nM}\\
& g_2(\mu)- g_2(\mu^{(k+1)}) + (\mu - \mu^{(k+1)})\left( -A_3^T(\lambda^k-\rho(A_1\theta+A_3\mu)) \right)\ge 0, \forall \mu \in \bR^{nM}\\
}
which is also the first order optimality condition for the regime:
\bas{
&\theta^{(k+1)} = \arg\min f(\theta) - \lambda^{(k)T}(A_1\theta) + \frac{\rho}{2}\|A_1\theta+A_2\eta+A_3\mu\|^2;\\
& (\eta^{(k+1)},\mu^{(k+1)}) = \arg\min_{\eta,\mu} g_1(\eta)+g_2(\mu) - \lambda^{(k)T}(A_2\eta-A_3\mu) + \frac{\rho}{2}\|A_1\theta^{(k+1)}+A_2\eta+A_3\mu\|^2;\\
& \begin{bmatrix}U^{(k+1)}\\V^{(k+1)}\end{bmatrix} = \begin{bmatrix}U^{(k)}\\V^{(k)}\end{bmatrix} -\rho(A_1\theta^{(k+1)}+A_2\eta^{(k+1)}+A_3\mu^{(k+1)})
}
Clearly, this is a specific application of the two-block ADMM by regarding $(\eta, \mu)$ as one variable, $B:=[A_2, A_3]$ as one matrix, and $g(\eta,\mu):=g_1(\eta)+g_2(\mu)$ as one function. 
Existing convergence results for the two-block ADMM thus hold for our case.

Now we note that $f$ is Lipschitz differentiable and strongly convex, both $f(\theta)$ and $g(\eta,\mu)$ are closed (their sublevel sets are closed) and proper (they neither take on the value $-\infty$ nor are they uniformly equal to $\infty$)).
Also $A$ and $B$ both have full column rank. Therefore by Theorem 7 in~\citet{nishihara2015general}, Algorithm~\ref{alg:palm} converges to the unique global solution of \eqref{eq:obj} linearly.
\end{proof}

\section{Proofs in Section~\ref{sec:consistency}}
\label{sec:proof_consistent}
\subsection{Proof of Theorem~\ref{th:consist}}
\begin{proof}[Proof of Theorem~\ref{th:consist}]
Define $H_X=I-H_v$ is the projection matrix projecting each voxel to $\text{span}(X)$.
Let $\hat{\theta}\in \bR^{nM}$ be the optimal solution for the following constrained optimization problem.
\bas{
\min_\theta \quad & \| \ttvec(Y^T)-\theta \|_F^2+\lambda \|D_v^T\theta\|_{\ell_1}\\
\mbox{s.t.} \quad & H_v\theta=0.
}
When $\text{rank}(X)=p$, by Lemma~\ref{lem:theta2gamma}, $\hat{\Gamma} = (X^TX)^{-1}X^T\text{mat}(\theta)_{n\times M}$. We first prove an oracle inequality for $\hat{\theta}$.
By the KKT condition, $\exists z\in \sign(D_v^T\hat{\theta}), \alpha\in \bR^{nM}$ such that 
\ba{
& 2(\hat{\theta}-y)+\lambda D_v^Tz+H_v\alpha=0 \label{eq:first_order}\\
& H_v\hat{\theta}=0 \nonumber
}
Multiplying $H_X$ on the left of Eq.~\eqref{eq:first_order}, we have $2H_X(\hat{\theta}-y)+\lambda H_XD_v^Tz =0$.
Equivalently, 
\bas{\forall \bar{\theta} \in R^{nM},\qquad 2\bar{\theta}^T H_X(\hat{\theta}-y)+\lambda \bar{\theta}^T H_XD_v^Tz =0}
Note that $H_X$ is the projection matrix to the column space of $\text{span}(X)$, it suffices to consider $\forall \bar{\theta}=\ttvec(\bar{\Gamma}^TX^T)$.
By definition of the sign operator, the following holds:
\bas{
&\hat{\theta}^T (\ttvec(Y^T)-\hat{\theta}) = \lambda\|D_v^T\hat{\theta}\|_{\ell_1}\\
&\bar{\theta}^T (\ttvec(Y^T)-\hat{\theta}) \le \lambda\|D_v^T\bar{\theta}\|_{\ell_1}
}
Subtracting the former from the latter, and replacing $\ttvec(Y^T)$ with $\theta^*+\epsilon$, we get
\bas{
(\bar{\theta}-\hat{\theta})^T (\theta^*-\hat{\theta})\le (\hat{\theta}-\bar{\theta})^T  \epsilon +\lambda\|D_v^T \bar{\theta}\|_{\ell_1}-\lambda\|D_v^T\hat{\theta}\|_{\ell_1}
}
Note $(\bar{\theta}-\hat{\theta})^T (\theta^*-\hat{\theta}) = \frac{1}{4}\left( \|\bar{\theta}-\hat{\theta}\|^2+\|\theta^*-\hat{\theta}\|^2 - \|\bar{\theta}-\theta^*\|^2 \right)$,
\ba{
\|\bar{\theta}-\hat{\theta}\|^2+\|\theta^*-\hat{\theta}\|^2 \le \|\bar{\theta}-\theta^*\|^2+ 4(\hat{\theta}-\bar{\theta})^T  \epsilon +4\lambda\|D_v^T\bar{\theta}\|_{\ell_1}-4\lambda\|D_v^T\hat{\theta}\|_{\ell_1}
\label{eq:diff_raw}
}

To bound $(\hat{\theta}-\bar{\theta})^T  \epsilon $, note $DD^T$ is the graph Laplacian, therefore when the graph is connected, we have $\text{ker}(D^T)=\text{ker}(DD^T)=\ttspan\{1_M\}$. Define $D^\dagger$ be the pseudo inverse of $D$, then $I-(D^\dagger)^T D^T$ is the projection matrix onto ker($D^T$),
\bas{
(\hat{\theta}-\bar{\theta})^T \epsilon = & \sum_{i=1}^n (\hat{\theta}_i -\bar{\theta}_i)^T \epsilon_i\\
= & \sum_{i=1}^n   ((I-(D^\dagger)^T D^T)\epsilon_i)^T(\hat{\theta}_i-\bar{\theta}_i)+((D^\dagger)^T D^T\epsilon_i)^T(\hat{\theta}_i-\bar{\theta}_i)\\
\le &  \sum_{i=1}^n \|(I-(D^\dagger)^T D^T)\epsilon_i\|\cdot \|\hat{\theta}_i -\bar{\theta}_i\|+\|(D^\dagger)^T\epsilon_i\|_\infty \cdot \|D^T(\hat{\theta}_i-\bar{\theta}_i)\|_{\ell_1}.
}
In view of the fact that $\epsilon_i\sim \cN(0,\sigma^2I_M)$ and $(I-(D^\dagger)^T D^T)$ being a projection matrix to a one dimensional space, by the tail bound for Gaussian random variables, $\forall i\in [n], \forall \delta>0$, 
\bas{
P(\|(I-(D^\dagger)^T D^T)\epsilon_i\|\ge 2\sigma\sqrt{2\log(2enM/\delta)} )\le \frac{\delta}{2n}
}
For the second part, by the maximal inequality for Gaussian random variables (\cite{massart2007concentration} Thm 3.12), and the variance of elements of $(D^\dagger)^T)\epsilon_i$ is upper bounded by $\rho^2\sigma^2$,
\bas{
P(\|(D^\dagger)^T\epsilon_i\|_\infty \ge \rho\sigma\sqrt{2\log(2emnM/\delta)} )\le \frac{\delta}{2n}
}
Applying union bound, with probability at least $1-\delta$,
\beq{
\bsplt{
(\hat{\theta}-\bar{\theta})^T \epsilon \le& \sum_{i=1}^n \left( 2\sigma\sqrt{2\log(2enM/\delta)}\|\hat{\theta}_i -\bar{\theta}_i\|+ \rho\sigma\sqrt{2\log(2emnM/\delta)}\|D^T(\hat{\theta}_i-\bar{\theta}_i)\|_{\ell_1} \right)\\
\le &  2\sigma\sqrt{2\log(2enM/\delta)} \|\hat{\theta} -\bar{\theta}\|+ \rho\sigma\sqrt{2\log(2emnM/\delta)}\left( \sum_{i=1}^n \|D^T(\hat{\theta}_i-\bar{\theta}_i)\|_{\ell_1} \right)
}
\label{eq:theta_eps}
}

By the triangle inequality, 
\bas{
\|D^T(\that-\tbar)_{T^c}\|_{\ell_1}-\|D^T\that_{T^c}\|_{\ell_1}\le \|D^T\tbar_{T^c}\|_{\ell_1}\\
\|D^T\tbar_{T}\|_{\ell_1}-\|D^T\that_{T}\|_{\ell_1}\le \|D^T(\tbar-\that)_{T}\|_{\ell_1}
}
Hence
\bas{
&\|D^T(\hat{\theta}_i-\bar{\theta}_i)\|_{\ell_1}+\|D^T(\bar{\theta}_i)\|_{\ell_1}-\|D^T(\hat{\theta}_i)\|_{\ell_1} \\
=& \|D^T(\hat{\theta}_i-\bar{\theta}_i)_T\|_{\ell_1}+\|D^T(\hat{\theta}_i-\bar{\theta}_i)_{T^c}\|_{\ell_1}+\|D^T(\bar{\theta}_i)_T\|_{\ell_1}+\|D^T(\bar{\theta}_i)_{T^c}\|_{\ell_1}-\|D^T(\hat{\theta}_i)_T\|_{\ell_1}-\|D^T(\hat{\theta}_i)_{T^c}\|_{\ell_1} \\
\le& 2\|D^T(\hat{\theta}_i-\bar{\theta}_i)_T\|_{\ell_1}+2\|D^T(\bar{\theta}_i)_{T^c}\|_{\ell_1}
}
By Definition 1 ,
$
\|D^T(\hat{\theta}_i-\bar{\theta}_i)_T\|_{\ell_1} \le \kappa_T^{-1}\sqrt{|T|}\|\hat{\theta}_i-\bar{\theta}_i\|.
$
We now plug above and \eqref{eq:theta_eps} back to \eqref{eq:diff_raw}, and take $\lambda = \rho\sigma\sqrt{\log(mnM/\delta)}$, then with probability at least $1-c_7n^{-1}$,
\beq{
\bsplt{
\|\bar{\theta}-\hat{\theta}\|^2+\|\theta^*-\hat{\theta}\|^2 \le & \|\bar{\theta}-\theta^*\|^2+ 8\sigma\sqrt{2\log(2en/\delta)} \|\hat{\theta} -\bar{\theta}\| + 4\lambda \|(D\bar{\theta})_{T^c}\|_{\ell_1} + 4\lambda \kappa_T^{-1}\sqrt{|T|}\|\hat{\theta}-\bar{\theta}\|}
\label{eq:tbar_minus_that}
}
Use Young's inequality, 
\bas{
8\sigma\sqrt{2\log(2en/\delta)} \|\hat{\theta} -\bar{\theta}\|\le & \frac{1}{2}\|\that-\tbar\|^2+64\sigma^2\log\left( \frac{2en}{\delta} \right)\\
4\lambda \kappa_T^{-1}\sqrt{|T|}\|\hat{\theta}-\bar{\theta}\| \le & \frac{1}{2}\|\that-\tbar\|^2+ 8\lambda^2\kappa_T^{-2}|T|
}
Canceling out $\|\that-\tbar\|^2$ on both sides of \eqref{eq:tbar_minus_that}, 
\bas{
\|\theta^*-\hat{\theta}\|^2 \le & \|\bar{\theta}-\theta^*\|^2+ 4\lambda \|(D_v^T\bar{\theta})_{T^c}\|_{\ell_1} +  64\sigma^2\log\left( \frac{2enM}{\delta} \right)+ 8\lambda^2\kappa_T^{-2}|T|
} 

Taking infimum on the right and plugging in $\lambda$ we have 
\bas{
&\|\theta^*-\hat{\theta}\|^2\le \\
 & \inf_{\bar{\theta}\in \bR^{n M}: H_X\bar{\theta}=\bar{\theta}} \left\{ \|\bar{\theta}-\theta^*\|^2+ 4\lambda \|(D_v^T\bar{\theta})_{T^c}\|_{\ell_1} \right\} + 64\sigma^2\log\left( \frac{2enM}{\delta} \right)+ 8\rho^2\sigma^2\log\left(\frac{mnM}{\delta}\right)\kappa_T^{-2}|T|
}
\label{eq:bound_theta}
\end{proof}

\subsection{Proof of Corollary~\ref{cor:gamma_oracle} }
Corollary~\ref{cor:gamma_oracle} can be proved by combining Lemma~\ref{lem:theta2gamma} and Theorem~\ref{th:consist}.

\begin{proof}[Proof of Corollary~\ref{cor:gamma_oracle}]
When $\frac{1}{n}X^TX=I_p$, we will be able to control the error in $\Gamma$, notice that $\|\ttvec(A)\|_2^2 = \|A\|_F^2$, we have from \eqref{eq:bound_theta} that 
\bas{
& \|\hat{\Gamma}-\Gamma^*\|_F^2 = \tr((X^TX)^{-1}X^T\text{mat}((\theta^*-\that) (\theta^*-\that)^T)X(X^TX)^{-1}) \\
= & \frac{1}{n} \|\theta^*-\hat{\theta}\|^2\\
=& \inf_{\bar{\theta}\in \bR^{n M}: H_X\bar{\theta}=\bar{\theta}} \left( \frac{1}{n} \|\bar{\theta}-\theta^*\|^2+ \frac{4\lambda}{n} \|(D_v^T\bar{\theta})_{T^c}\|_{\ell_1}  \right) + \frac{1}{n}\left( 4\sigma\sqrt{2\log(2enM/\delta)}+2\lambda \kappa_T^{-1}\sqrt{|T|}\right)^2\\
\le &  \inf_{\bar{\Gamma}\in \bR^{p\times M}} \left( \|\bar{\Gamma}-\hat{\Gamma}\|_F^2+ \frac{4\lambda}{ n} \|(X\bar{\Gamma}D)_{T^c}\|_{\ell_1} \right) + \frac{1}{n }\left( 4\sigma\sqrt{2\log(2enM/\delta)}+2\lambda \kappa_T^{-1}\sqrt{|T|}\right)^2
}
\end{proof}
\bibliography{fused}

\begin{thebibliography}{35}
\providecommand{\natexlab}[1]{#1}
\providecommand{\url}[1]{\texttt{#1}}
\providecommand{\urlprefix}{}

\bibitem[{Naslund et~al.(2015)Naslund, John A and Aschbrenner, Kelly A and
  Barre, Laura K and Bartels, Stephen J}]{naslund2015feasibility}
Naslund JA, Aschbrenner KA, Barre LK, Bartels SJ.
\newblock Feasibility of popular m-health technologies for activity tracking
  among individuals with serious mental illness.
\newblock Telemedicine and e-Health 2015;21(3):213--216.

\bibitem[{Pantelopoulos and Bourbakis(2010)Pantelopoulos, Alexandros and
  Bourbakis, Nikolaos G}]{pantelopoulos2010survey}
Pantelopoulos A, Bourbakis NG.
\newblock A survey on wearable sensor-based systems for health monitoring and
  prognosis.
\newblock IEEE Transactions on Systems, Man, and Cybernetics, Part C
  (Applications and Reviews) 2010;40(1):1--12.

\bibitem[{Patel et~al.(2015)Patel, Mitesh S and Asch, David A and Volpp, Kevin
  G}]{patel2015wearable}
Patel MS, Asch DA, Volpp KG.
\newblock Wearable devices as facilitators, not drivers, of health behavior
  change.
\newblock Jama 2015;313(5):459--460.

\bibitem[{Rudin et~al.(1992)Rudin, Leonid I and Osher, Stanley and Fatemi,
  Emad}]{rudin1992nonlinear}
Rudin LI, Osher S, Fatemi E.
\newblock Nonlinear total variation based noise removal algorithms.
\newblock Physica D: Nonlinear Phenomena 1992;60(1):259--268.

\bibitem[{Steidl et~al.(2006)Steidl, Gabriele and Didas, Stephan and Neumann,
  Julia}]{steidl2006splines}
Steidl G, Didas S, Neumann J.
\newblock Splines in higher order TV regularization.
\newblock International journal of computer vision 2006;70(3):241--255.

\bibitem[{Beck and Teboulle(2009)Beck, Amir and Teboulle, Marc}]{beck2009fast}
Beck A, Teboulle M.
\newblock Fast gradient-based algorithms for constrained total variation image
  denoising and deblurring problems.
\newblock IEEE Transactions on Image Processing 2009;18(11):2419--2434.

\bibitem[{Condat(2013)Condat, Laurent}]{condat2013direct}
Condat L.
\newblock A direct algorithm for 1D total variation denoising.
\newblock IEEE Signal Processing Letters 2013;20(11):1054--1057.

\bibitem[{Padilla et~al.(2016)Padilla, Oscar Hernan Madrid and Scott, James G
  and Sharpnack, James and Tibshirani, Ryan J}]{padilla2016dfs}
Padilla OHM, Scott JG, Sharpnack J, Tibshirani RJ.
\newblock The DFS fused lasso: Linear-time denoising over general graphs.
\newblock arXiv preprint arXiv:160803384 2016;.

\bibitem[{Chen et~al.(2016)Chen, Yao and Wang, Xiao and Kong, Linglong and Zhu,
  Hongtu}]{chen2016local}
Chen Y, Wang X, Kong L, Zhu H.
\newblock Local Region Sparse Learning for Image-on-Scalar Regression.
\newblock arXiv preprint arXiv:160508501 2016;.

\bibitem[{Tibshirani et~al.(2005)Tibshirani, Robert and Saunders, Michael and
  Rosset, Saharon and Zhu, Ji and Knight, Keith}]{tibshirani2005}
Tibshirani R, Saunders M, Rosset S, Zhu J, Knight K.
\newblock Sparsity and smoothness via the fused lasso.
\newblock Journal of the Royal Statistical Society: Series B (Statistical
  Methodology) 2005;67(1):91--108.

\bibitem[{Tibshirani and Taylor(2011)Ryan J. Tibshirani and Jonathan
  Taylor}]{tib2011}
Tibshirani RJ, Taylor J.
\newblock The Solution Path of the Generalized Lasso.
\newblock Annals of Statistics 2011;39(3):1335--1371.

\bibitem[{Wang et~al.(2017)Xiao Wang and Hongtu Zhu and for the Alzheimer’s
  Disease Neuroimaging Initiative}]{wang2018}
Wang X, Zhu H, for~the Alzheimer’s Disease Neuroimaging~Initiative.
\newblock Generalized Scalar-on-Image Regression Models via Total Variation.
\newblock Journal of the American Statistical Association
  2017;112(519):1156--1168.

\bibitem[{Grosenick et~al.(2013)Logan Grosenick and Brad Klingenberg and Kiefer
  Katovich and Brian Knutson and Jonathan E. Taylor}]{logan2013}
Grosenick L, Klingenberg B, Katovich K, Knutson B, Taylor JE.
\newblock Interpretable whole-brain prediction analysis with GraphNet.
\newblock NeuroImage 2013;72:304 -- 321.
\newblock
  \urlprefix\url{http://www.sciencedirect.com/science/article/pii/S1053811912012487}.

\bibitem[{Reiss et~al.(2010)Reiss, Philip T and Huang, Lei and Mennes,
  Maarten}]{reiss2010fast}
Reiss PT, Huang L, Mennes M.
\newblock Fast function-on-scalar regression with penalized basis expansions.
\newblock International Journal of Biostatistics 2010;6(1).

\bibitem[{Goldsmith and Kitago(2016)Goldsmith, Jeff and Kitago,
  Tomoko}]{goldsmith2016assessing}
Goldsmith J, Kitago T.
\newblock Assessing systematic effects of stroke on motor control by using
  hierarchical function-on-scalar regression.
\newblock Journal of the Royal Statistical Society: Series C (Applied
  Statistics) 2016;65(2):215--236.

\bibitem[{Scheipl et~al.(2015)Scheipl, Fabian and Staicu, Ana-Maria and Greven,
  Sonja}]{scheipl2015functional}
Scheipl F, Staicu AM, Greven S.
\newblock Functional additive mixed models.
\newblock Journal of Computational and Graphical Statistics
  2015;24(2):477--501.

\bibitem[{Zhou et~al.(2013)Zhou, Hua and Li, Lexin and Zhu,
  Hongtu}]{zhou2013tensor}
Zhou H, Li L, Zhu H.
\newblock Tensor regression with applications in neuroimaging data analysis.
\newblock Journal of the American Statistical Association
  2013;108(502):540--552.

\bibitem[{Li and Zhang(2017)Lexin Li and Xin Zhang}]{li2017parsimonious}
Li L, Zhang X.
\newblock Parsimonious Tensor Response Regression.
\newblock Journal of the American Statistical Association
  2017;112(519):1131--1146.

\bibitem[{Cook et~al.(2010)R. Dennis Cook and Bing Li and Francesca
  Chiaromonte}]{Cook2010}
Cook RD, Li B, Chiaromonte F.
\newblock Envelope Models for Parsimonious and Efficient Multivariate Linear
  Regression.
\newblock Statistica Sinica 2010;20(3):927--960.

\bibitem[{Boyd et~al.(2011)Boyd, Stephen and Parikh, Neal and Chu, Eric and
  Peleato, Borja and Eckstein, Jonathan}]{boyd2011distributed}
Boyd S, Parikh N, Chu E, Peleato B, Eckstein J.
\newblock Distributed optimization and statistical learning via the alternating
  direction method of multipliers.
\newblock Foundations and Trends{\textregistered} in Machine Learning
  2011;3(1):1--122.

\bibitem[{Wahlberg et~al.(2012)Wahlberg, Bo and Boyd, Stephen and Annergren,
  Mariette and Wang, Yang}]{wahlberg2012admm}
Wahlberg B, Boyd S, Annergren M, Wang Y.
\newblock An ADMM algorithm for a class of total variation regularized
  estimation problems.
\newblock IFAC Proceedings Volumes 2012;45(16):83--88.

\bibitem[{Davies and Kovac(2001)Davies, P Laurie and Kovac,
  Arne}]{davies2001local}
Davies PL, Kovac A.
\newblock Local extremes, runs, strings and multiresolution.
\newblock Annals of Statistics 2001;p. 1--48.

\bibitem[{Johnson(2013)Johnson, Nicholas A}]{johnson2013dynamic}
Johnson NA.
\newblock A dynamic programming algorithm for the fused lasso and l
  0-segmentation.
\newblock Journal of Computational and Graphical Statistics
  2013;22(2):246--260.

\bibitem[{W.~Tansey and Scott.(2017)W. Tansey, O. Koyejo, R. A. Poldrack, and
  J. G. Scott.}]{Tansey2017}
W~Tansey RAP O~Koyejo, Scott JG.
\newblock "False Discovery Rate Smoothing," Supplementary material.
\newblock Journal of the Amerian Statistical Association (JASA): Theory and
  Methods 2017;.

\bibitem[{Nishihara et~al.(2015)Nishihara, Robert and Lessard, Laurent and
  Recht, Ben and Packard, Andrew and Jordan, Michael}]{nishihara2015general}
Nishihara R, Lessard L, Recht B, Packard A, Jordan M.
\newblock A General Analysis of the Convergence of ADMM.
\newblock In: International Conference on Machine Learning; 2015. p. 343--352.

\bibitem[{Van De~Geer et~al.(2009)Van De Geer, Sara A and B{\"u}hlmann, Peter
  and others}]{van2009conditions}
Van De~Geer SA, B{\"u}hlmann P, et~al.
\newblock On the conditions used to prove oracle results for the Lasso.
\newblock Electronic Journal of Statistics 2009;3:1360--1392.

\bibitem[{H{\"u}tter and Rigollet(2016)H{\"u}tter, Jan-Christian and Rigollet,
  Philippe}]{hutter2016optimal}
H{\"u}tter JC, Rigollet P.
\newblock Optimal rates for total variation denoising.
\newblock arXiv preprint arXiv:160309388 2016;.

\bibitem[{Merikangas et~al.(2014)Merikangas, Kathleen R. and Cui, Lihong and
  Heaton, Leanne. and Nakamura, Erin F. and Roca, C and Ding, J. and Qin, H.
  and Guo, Wei and Shugart, Y. Y. and Yao-Shugart, Y. and Zarate, C. and Angst,
  Jules}]{Merikangas2014}
Merikangas KR, Cui L, Heaton L, Nakamura EF, Roca C, Ding J, et~al.
\newblock {I}ndependence of familial transmission of mania and depression:
  results of the {N}{I}{M}{H} family study of affective spectrum disorders.
\newblock Molecular Psychiatry 2014;19(2):214--219.

\bibitem[{Shou et~al.(2017)Shou, Haochang and Cui, Lihong and Hickie, Ian and
  Lameira, Diana and Lamers, Femake and Zhang, Jihui and Crainiceanu, Ciprian
  and Zipunnikov, Vadim and Merikangas, Kathleen R.}]{Shou2017}
Shou H, Cui L, Hickie I, Lameira D, Lamers F, Zhang J, et~al.
\newblock {D}ysregulation of objectively assessed 24-hour motor activity
  patterns as a potential marker for bipolar {I} disorder: results of a
  community-based family study.
\newblock Translational Psychiatry 2017;7(8):e1211.

\bibitem[{Bohland et~al.(2009)Bohland, J. W. and Bokil, H. and Allen, C. B. and
  Mitra, P. P.}]{Bohland2009}
Bohland JW, Bokil H, Allen CB, Mitra PP.
\newblock {{T}he brain atlas concordance problem: quantitative comparison of
  anatomical parcellations}.
\newblock PLoS ONE 2009 Sep;4(9):e7200.

\bibitem[{van Rooij et~al.(2015)van Rooij, Daan and Hartman, Catharina A. and
  Mennes, Maarten and Oosterlaan, Jaap and Franke, Barbara and Rommelse, Nanda
  and Heslenfeld, Dirk and Faraone, Stephen V. and Buitelaar, Jan K. and
  Hoekstra, Pieter J.}]{vanRooij2015}
van Rooij D, Hartman CA, Mennes M, Oosterlaan J, Franke B, Rommelse N, et~al.
\newblock {{A}ltered neural connectivity during response inhibition in
  adolescents with attention-deficit/hyperactivity disorder and their
  unaffected siblings}.
\newblock Neuroimage Clinical 2015;7:325--335.

\bibitem[{Higo et~al.(2011)Higo, Takayasu and Mars, Rogier B. and Boorman, Erie
  D. and Buch, Ethan R. and Rushworth, Matthew F.S.}]{Higo2011}
Higo T, Mars RB, Boorman ED, Buch ER, Rushworth MFS.
\newblock {{D}istributed and causal influence of frontal operculum in task
  control}.
\newblock Proceedings of National Academy of Science 2011
  Mar;108(10):4230--4235.

\bibitem[{Cuzzocreo et~al.(2009)Cuzzocreo, Jennifer L. and Yassa, Michael A.
  and Verduzco, Guillermo and Honeycutt, Nancy A. and Scott, David J. and
  Bassett, Susan S.}]{cuzzocreo2009}
Cuzzocreo JL, Yassa MA, Verduzco G, Honeycutt NA, Scott DJ, Bassett SS.
\newblock {{E}ffect of handedness on f{M}{R}{I} activation in the medial
  temporal lobe during an auditory verbal memory task}.
\newblock Human Brain Mapping 2009 Apr;30(4):1271--1278.

\bibitem[{Chen et~al.(2016)Chen, Caihua and He, Bingsheng and Ye, Yinyu and
  Yuan, Xiaoming}]{chen2016direct}
Chen C, He B, Ye Y, Yuan X.
\newblock The direct extension of ADMM for multi-block convex minimization
  problems is not necessarily convergent.
\newblock Mathematical Programming 2016;155(1-2):57--79.

\bibitem[{Massart(2007)Massart, Pascal}]{massart2007concentration}
Massart P.
\newblock Concentration inequalities and model selection, vol.~6.
\newblock Springer; 2007.

\end{thebibliography}
\clearpage


\end{document}